\documentclass{llncs}

\usepackage{amsfonts,amsmath,amssymb}
\usepackage[mathscr]{euscript}
\usepackage{hyperref,paralist}
\usepackage{todonotes}
\usepackage{etex}
\usepackage{microtype}

\usepackage[linesnumbered,ruled,vlined]{algorithm2e}
\SetKw{KwLet}{let}
\SetKw{KwAbort}{abort}
\SetKw{KwChoose}{choose}
\SetKwFor{Forall}{for all}{do}{end}
\SetKwFor{Func}{function}{}{}
\usepackage{listings}
\usepackage{pstricks,pst-node,pst-plot,pstricks-add}

\newcommand{\rronly}[1]{#1}
\newcommand{\pponly}[1]{}

\let\oldfootnotesize\footnotesize
\renewcommand*{\footnotesize}{\pponly{\scriptsize}}

\newcommand{\negskip}{\pponly{\vspace*{-1.5ex}}}

\newcommand{\init}{I}
\newcommand{\trans}{T}
\newcommand{\sset}{\Sigma}
\newcommand{\iset}{\Upsilon}

\newcommand{\eg}{e.g.}
\newcommand{\ie}{i.e.}
\newcommand{\sep}{$\vee$}

\newcommand{\nat}{\mathbb{N}}

\newcommand{\sinit}{I}
\newcommand{\sfinal}{F}
\newcommand{\assume}{\varphi}
\newcommand{\assert}{\psi}
\newcommand{\chain}{\chi}
\newcommand{\true}{\mathit{true}}
\newcommand{\false}{\mathit{false}}

\newcommand{\chpath}{\mathit{CheckPath}}
\newcommand{\chkreach}{\mathit{checkKreach}}
\newcommand{\getkedges}{\mathit{GetKreachEdges}}
\newcommand{\failedpath}{\mathit{failed\_path}}

\renewcommand{\paragraph}[1]{\noindent\textbf{#1}.~}

\title{Chaining Test Cases\\
for Reactive System Testing\thanks{Supported by
the EU FP7 STREP PINCETTE, the ARTEMIS VeTeSS project, and
ERC project 280053.}\rronly{\\{\normalsize(extended version)}}}
\author{Peter Schrammel \and Tom Melham \and Daniel Kroening}
\institute{University of Oxford\\ Department of Computer Science\\ \texttt{first.lastname@cs.ox.ac.uk}}

\begin{document}
\maketitle

\begin{abstract}

Testing of synchronous reactive systems is challenging because long input
sequences are often needed to drive them into a state to test a
desired feature.  This is particularly problematic in
\textit{on-target testing}, where a system is tested in its real-life
application environment and the amount of time required for resetting is high.
This paper presents an approach to discovering a \textit{test case
chain}---a single software execution that covers a group of test goals and
minimises overall test execution time.  Our technique targets the scenario
in which test goals for the requirements are given as safety properties.  We
give conditions for the existence and minimality of a single test case chain
and minimise the number of test case chains if a single test case chain is infeasible.
We report experimental results with a prototype tool for C code generated from
\textsc{Simulink} models and compare it to state-of-the-art test suite
generators.
\end{abstract}

% \begin{keywords}
%  reactive systems, test case generation, model checking
% \end{keywords}

%===============================================================================
\section{Introduction}\label{sec:intro}
%===============================================================================

Safety-critical embedded software, e.g., in the automotive or avionics
domain, is often implemented as a \emph{synchronous reactive system}.
These systems
compute their new state and their output as
functions of old state and the given inputs.  As these systems frequently
have to satisfy high safety standards, tool support for systematic testing
is highly desirable. The completeness of the testing process is frequently
measured by defining a set of \emph{test goals}, which are typically
formulated as reachability properties. A good-quality test suite is a set
of input sequences that drive the system into states that cover
a large fraction of those goals.

Test suites generated by random test generators often contain a huge
number of redundant test cases.
Directed test case generation often
requires lengthy input sequences to drive the system into a state where
the desired feature can be tested.
Furthermore, to execute the test suite, test cases must be chained
manually or the system must be reset after
executing each test case.  This is a serious problem in \textit{on-target
testing}, where a system is tested in its real-life application environment
and resetting might be very time-consuming~\cite{HU10}.

\nopagebreak

This paper presents an approach to discovering a \textit{test case
chain}---a single test case that covers a set of multiple test goals and
minimises overall test execution time.
The essence of the problem is to find a shortest path through the system
that covers all the test goals.

%-------------------------------------------------------------------------------
\paragraph{Example}
To illustrate the problem and our approach, we reuse the classical
cruise controller example given in \cite{Bos07}. There are five Boolean inputs,
two for actuation of the \emph{gas} and \emph{brake}
pedals, a toggle \emph{button} to enable the cruise control, and two
sensors indicating whether the car is \emph{acc}- or
\emph{dec}elerating. There are three state variables: \emph{speed},
\emph{enable}, which is true when cruise control is enabled, and
\emph{mode} indicating whether cruise control is turned
\emph{OFF}, actually active (\emph{ON}), or temporarily inactive, \ie,
\emph{DIS}engaged while user pushes the gas or brake pedal.
A C implementation, with the
structure typical of code generated from \textsc{Simulink} models, is
given in Fig.~\ref{fig:code-ex-cruise} and its state machine is
depicted in Fig.~\ref{fig:ex-cruise}. 
The function \texttt{compute} is executed periodically (\eg~on a timer
interrupt). Thus, there is
a notion of \emph{step} that relates to execution time.
%
%%%%%%%%%%%%%%%%%%%%%%%%%%%%%%%%%%%%%%%%%%%%%%%%%%%%%%%%%%%%%%%%%%%%%%%%%%%%%%%%
\begin{figure}[t]
{\vspace*{-1ex} \scriptsize 
\begin{lstlisting}
void init(t_state *s) { s->mode = OFF; s->speed = 0; s->enable = FALSE; }
void compute(t_input *i, t_state *s) {
  mode = s->mode;
  switch(mode) {
    case ON: if(i->gas || i->brake) s->mode=DIS; break;
    case DIS: 
      if( (s->speed==2 && (i->dec || i->brake)) || (s->speed==0 && (i->acc || i->gas)) ) 
        s->mode=ON;
      break;
    case OFF: 
      if( s->speed==0 && s->enable && (i->gas || i->acc) ||
          s->speed==1 && i->button ||
          s->speed==2 && s->enable && (i->brake || i->dec) ) 
        s->mode=ON;
      break;
  }
  if(i->button) s->enable = !s->enable;
  if((i->gas || mode!=ON && i->acc) && s->speed<2) s->speed++;
  if((i->brake || mode!=ON && i->dec) && s->speed>0) s->speed--;
}
\end{lstlisting}
}
\vspace*{-2ex}
\caption{Code generated for cruise controller
example\label{fig:code-ex-cruise}}
\end{figure}
%%%%%%%%%%%%%%%%%%%%%%%%%%%%%%%%%%%%%%%%%%%%%%%%%%%%%%%%%%%%%%%%%%%%%%%%%%%%%%%%
%

We formulate some LTL properties for which we want to
generate test cases:\\
\begin{tabular}{@{\hspace{2em}}l@{:}@{\hspace{0.5em}}l}
$p_1$ & $\mathbf{G}\big(mode=\mathit{ON}  \wedge speed=1 \wedge dec \Rightarrow \mathbf{X}(speed=1)\big)$ \\
$p_2$ & $\mathbf{G}\big(mode=\mathit{DIS} \wedge speed=2 \wedge dec \Rightarrow \mathbf{X}(mode=ON)\big)$ \\
$p_3$ & $\mathbf{G}\big(mode=\mathit{ON}  \wedge brake \Rightarrow \mathbf{X}(mode=\mathit{DIS})\big)$ \\
$p_4$ & $\mathbf{G}\big(mode=\mathit{OFF} \wedge speed=2 \wedge \neg enable \wedge button \Rightarrow \mathbf{X}~enable\big)$ 
\end{tabular}

We observe that each of the properties above relates to a particular
transition in the state machine (shown as bold edge labels in
Fig.~\ref{fig:ex-cruise}).  A \emph{test case} is a sequence of inputs that
determines a (bounded) execution path through the system.  The \emph{length}
of a test case is the length of this sequence.  A test case \emph{covers a
property} if it triggers the transition the property relates to.  A
\emph{test suite} is a set of test cases that covers all the properties.

Ideally, we can obtain a single test case that covers all
properties in a single execution.  We call a test case that covers a
sequence of properties a \emph{test case chain}.
Our goal is to synthesise minimal test case chains---test case chains with
fewest transitions. It is not always possible
to generate a single test case chain that covers all properties;
multiple test case chains may be required.
%
%%%%%%%%%%%%%%%%%%% example cruise control %%%%%%%%%%%%%%%%%%%%%%%%%%%%%%%%%%%%%
\begin{figure}[t]
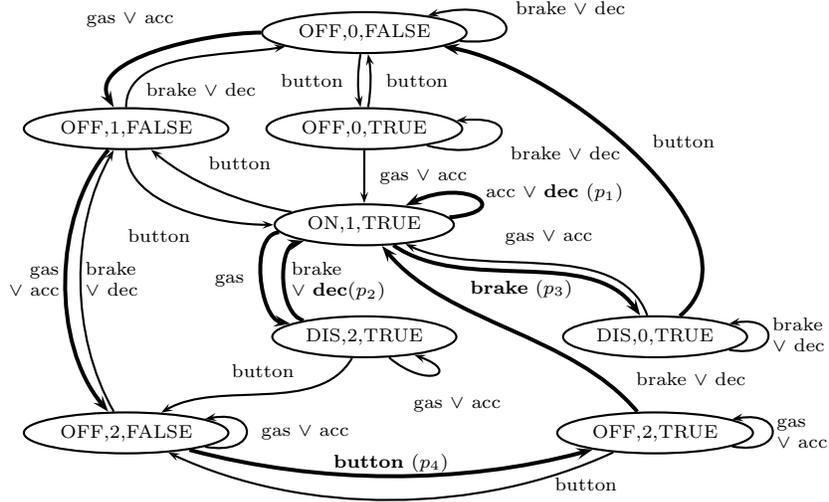

\scriptsize
\centering
\psset{arrows=->}
\begin{tabular}{@{\hspace{-3em}}c@{\hspace{1.5em}}c@{\hspace{4.2em}}c@{}}
 & \ovalnode{n000}{OFF,0,FALSE} \\[8ex]
\ovalnode{n010}{OFF,1,FALSE} & \ovalnode{n001}{OFF,0,TRUE} \\[8ex]
 & \ovalnode{n111}{ON,1,TRUE} \\[10ex]
 & \ovalnode{n221}{DIS,2,TRUE} & \ovalnode{n201}{DIS,0,TRUE} \\[8ex]
\ovalnode{n020}{OFF,2,FALSE} && \ovalnode{n021}{OFF,2,TRUE}
\end{tabular}
\nccurve[angleA=-5,angleB=15,ncurv=4]{n000}{n000}\nbput{brake \sep~dec}
\nccurve[angleA=-15,angleB=5,ncurv=4]{n001}{n001}\nbput{brake \sep~dec}
\nccurve[angleA=-10,angleB=10,ncurv=4]{n201}{n201}\nbput{$\begin{array}{@{\hspace{-0.5em}}l}\text{brake}
     \\ $\sep$~\text{dec}\end{array}$}
\nccurve[angleA=-40,angleB=-20,ncurv=4]{n221}{n221}\nbput{gas \sep~acc}
\nccurve[angleA=-10,angleB=10,ncurv=4]{n020}{n020}\nbput{gas \sep~acc}
\nccurve[angleA=-10,angleB=10,ncurv=4]{n021}{n021}\nbput{$\begin{array}{@{\hspace{-0.5em}}l}\text{gas}
     \\ $\sep$~\text{acc}\end{array}$}
\nccurve[linewidth=1.5pt,angleA=5,angleB=25,ncurv=3]{n111}{n111}\nbput[npos=0.4,labelsep=0.1pt]{acc
  \sep~\textbf{dec} $(p_1)$}
\nccurve[linewidth=1.5pt,angleA=180,angleB=120]{n000}{n010}\nbput{gas \sep~acc}
\nccurve[angleA=90,angleB=-170]{n010}{n000}\nbput{brake \sep~dec}
\nccurve[angleA=-100,angleB=100]{n000}{n001}\nbput{button}
\nccurve[angleA=80,angleB=-80]{n001}{n000}\nbput{button}
\nccurve[angleA=-90,angleB=180]{n010}{n111}\nbput{button}
\nccurve[angleA=170,angleB=-40]{n111}{n010}\nbput{button}
\nccurve[linewidth=1.5pt,angleA=-15,angleB=-165]{n020}{n021}\naput[labelsep=0.1pt]{\textbf{button}
    $(p_4)$}
\nccurve[angleA=-155,angleB=-25]{n021}{n020}\naput[labelsep=0.1pt,npos=0.1]{button}
\nccurve[angleA=-120,angleB=30]{n221}{n020}\nbput{button}
\nccurve[linewidth=1.5pt,angleA=40,angleB=-10]{n201}{n000}\nbput{button}
\nccurve[angleA=-90,angleB=90]{n001}{n111}\naput{gas \sep~acc}
\nccurve[angleA=110,angleB=-25]{n201}{n111}\nbput{gas \sep~acc}
\nccurve[linewidth=1.5pt,angleA=-35,angleB=125]{n111}{n201}\nbput[labelsep=0.1pt]{\textbf{brake}
   $(p_3)$}
\nccurve[linewidth=1.5pt,angleA=-130,angleB=130]{n010}{n020}\nbput{$\begin{array}{r@{\hspace{-0.5em}}}\text{gas}
     \\ $\sep$~\text{acc}\end{array}$}
\nccurve[linewidth=1.5pt,angleA=-175,angleB=170]{n111}{n221}\nbput{gas}
\nccurve[angleA=120,angleB=-120]{n020}{n010}\nbput{$\begin{array}{@{\hspace{-0.5em}}l}\text{brake}
     \\ $\sep$~\text{dec}\end{array}$}
\nccurve[linewidth=1.5pt,angleA=165,angleB=-165]{n221}{n111}\nbput{$\begin{array}{@{\hspace{-0.5em}}l}\text{brake}
     \\ $\sep$~\textbf{dec} (p_2)\end{array}$}
\nccurve[linewidth=1.5pt,angleA=130,angleB=-50]{n021}{n111}\nbput[npos=0.06]{brake \sep~dec}
\vspace*{5ex}
\caption{\label{fig:ex-cruise} State machine of the example.
  Edges are labelled by inputs and nodes
  by state  $\langle \mathit{mode},\mathit{speed},\mathit{enable}\rangle$.
  Properties are in bold, bold edges show a minimal test case chain.}
\end{figure}
%%%%%%%%%%%%%%%%%%%%%%%%%%%%%%%%%%%%%%%%%%%%%%%%%%%%%%%%%%%%%%%%%%%%%%%%%%%%%%%%

We compute such a minimal test case chain from a set of start states $\sinit$ via
a set of given properties $P=\{p_1,p_2,\ldots\}$ to a set of final states
$\sfinal$.  For our example, with $\sinit = \sfinal = \{mode=\mathit{OFF} \wedge
speed=0 \wedge \neg enable\}$ and $P=\{p_1, p_2, p_3, p_4\}$, for instance,
we obtain the test case chain consisting of the bold edges in
Fig.~\ref{fig:ex-cruise}.  First, this chain advances to $p_4$, then covers
$p_1$, $p_2$, and $p_3$, and finally goes to $\sfinal$.  One can assert that
this path has the minimal length of 9~steps.

Testing problems similar to ours have been addressed by research on
\emph{minimal checking sequences} in
conformance
testing~\pponly{\cite{NMHN13,PSY12,HU10}}\rronly{\cite{NMHN13,PSY12,HU10,HU06,Hie04}}.
This work analyses~automata-based specifications
that encode system control and have transitions labelled with
operations on data variables.
The challenge here is to find short transition paths based
on a given coverage criterion that are feasible, \ie~consistent with
the data operations.  Random test case generation can then be used to
discover such a path.
In contrast, our approach analyses the code generated from models or
the implementation code itself, and it can handle partial
specifications expressed as a collection of safety properties.  A
common example is acceptance testing in the automotive domain.  Our
solution uses bounded model checking to generate test cases guaranteed
to exercise the desired functionality.

%-------------------------------------------------------------------------------
\paragraph{Contributions}
The contributions of this paper can be summarised as follows:
\begin{compactitem}
\item We present a new algorithm to compute minimal test
  chains that first constructs a weighted digraph
  abstraction using a reachability analysis, on which the minimisation
  is performed as a second step. The final step is to compute the test input
  sequence. We give conditions for the existence and
  minimality of a single test case chain and propose algorithms to
  handle the general case.
\item We have implemented a tool, \textsc{ChainCover}\rronly{\footnote{\url{http://www.cprover.org/chaincover/}}}, for C
  code generated from \textsc{Simulink} models, on top of the
  \textsc{Cbmc} bounded model checker and the \textsc{Lkh}
  travelling salesman problem solver.
\item We present experimental results to demonstrate that our approach is
  viable on a set of benchmarks, mainly from automotive industry, and is
  more efficient than state-of-the-art test suite generators.
\end{compactitem}

%===============================================================================
\section{Preliminaries}\label{sec:prelim}
%===============================================================================

%-------------------------------------------------------------------------------
\paragraph{Program model}
A \emph{program} is given by $(\sset,\iset,\trans,\init)$ with
finite sets of states $\sset$ and inputs $\iset$, a transition
relation $\trans \subseteq (\sset\times\iset\times\sset)$, and a set of
initial states $\sinit\subseteq\sset$.
An \emph{execution} of a program is a (possibly) infinite sequence of transitions 
$s_0\xrightarrow{i_0}  s_1 \xrightarrow{i_1 } s_2 \rightarrow\ldots$
with $s_0\in\init$ and for all $k\geq 0$, $(s_k,i_k,s_{k+1})\in\trans$.

%-------------------------------------------------------------------------------
\paragraph{Properties}
We consider specifications given as a set of safety properties
$P=\{p_1,\ldots,p_{|P|}\}$. 
The properties are given as a formula over
state variables $s$ and input variables $i$ and are of the form
$\mathbf{G}\big(\assume\Rightarrow \assert
\big)$ where $\assume$ describes an \emph{assumption} and $\assert$ is
the \emph{assertion} to be checked.  
$\assume$ specifies a test goal, whereas $\assert$ defines the test
outcome; hence, for test case generation, only $\assume$ is needed.
We denote by $\Pi$ the set of property assumptions.
$\assume$ is a temporal logic formula
built using the operators $\wedge,\vee,\neg,\mathbf{X}$, \ie, 
it describes sets of finite paths.
An execution $\pi=\langle s_0,s_1,\ldots\rangle$ \emph{covers} a
property iff it contains a subpath $\langle s_k,\ldots s_{k+j}\rangle$
that satisfies~$\assume$
($j$ is the nesting depth of $\mathbf{X}$ operators in $\assume$), 
\ie, 
\rronly{$}$\exists k\geq 0: \exists i_k,\ldots,i_{k+j}:
\assume(s_k,i_k,\ldots,s_{k+j},i_{k+j}) \wedge \bigwedge_{k\leq m\leq
  k+j} T(s_m,i_m,s_{m+1}).$\rronly{$}
We call the set of states $s_k$ satisfying $\assume$ the \emph{trigger}
$\widehat{\assume}$ of the property.

For our method, it is not essential whether $\assume$ describes a set
of paths or just a set of states; thus, to simplify the presentation,
we assume that the property assumptions do not contain $\mathbf{X}$
operators.
Single-step transition properties
$\mathbf{G}\big(\assume\Rightarrow \mathbf{X}\assert\big)$
fall into this category, for example.
In this case, $\assume$ is equivalent to its trigger $\widehat{\assume}$.

Moreover, we assume that property assumptions are non-overlapping,
\ie~the sub-paths satisfying the assumptions do not share any edges.
Our minimality results only apply to such specifications.  Detecting
overlappings is a hard problem~\cite{BU91} that goes beyond the scope
of this paper.

%-------------------------------------------------------------------------------
\paragraph{Test cases}
A \emph{test case} is an input sequence $\langle
i_0,\ldots,i_n\rangle$ and generates an execution $\pi=\langle
s_0,\ldots,s_{n+1}\rangle$.
A test case \emph{covers} a property $p$ iff its execution covers the property.

%===============================================================================
\section{Chaining Test Cases}\label{sec:chain}
%===============================================================================

%-------------------------------------------------------------------------------
\paragraph{The problem}
We are given a program $(\sset,\iset,\trans,\init)$, properties $P$, 
and a set of final states $\sfinal\subseteq\sset$.
A \emph{test case chain} \rronly{$\chain$}
% chi like chAiN
is a test case $\langle
i_0,\ldots,i_n\rangle$ that covers all properties in $P$,
\ie,~its execution $\langle s_0,\ldots,s_{n+1}\rangle$
starts in $s_0\in\sinit$, ends
in $s_{n+1}\in\sfinal$, and covers all properties in $P$.
A \emph{minimal test case chain} is a test case chain of minimal length.
The final states $\sfinal$ are used to ensure the
test execution ends in a desired state, \eg~``engines off'' or ``gear
locked in park~mode''.

%-------------------------------------------------------------------------------
\bigskip\paragraph{Our approach}
We now describe our basic algorithm, which has three steps:
\begin{compactenum}[(1)]
\item \emph{Abstraction:} We construct a \emph{property K-reachability
  graph} of the system. This is a weighted, directed graph with nodes
  representing the properties and edges labelled with the number of
  states through which execution must pass, up to length $K$, between
  the properties.
\item \emph{Optimisation:} We determine the shortest path that covers
  all properties in the abstraction. 
\item \emph{Concretisation:} Finally, we compute the corresponding concrete
  test case chain along the abstract path. % and repair
\end{compactenum}
We discuss the conditions under which we obtain the \emph{minimal} test case chain.
\pponly{
Due to space limitations, we refer to the
extended version \cite{SMK13b} for the pseudo-code of the
algorithms and the proofs omitted in this paper.}  
\rronly{This algorithm is given as Alg.~\ref{alg:minchain1}. 

%%%%%%%%%%%%%%%% algo minimal test case chain %%%%%%%%%%%%%%%%%%%%%%%%%%%%%%%%%%
\begin{algorithm}[ht]
\KwIn{program $(\sset,\iset,\trans,\init)$, properties
  $P$, formulas $\sinit$, $\sfinal$, reachability bound $K$}
\KwOut{test case chain $\chain=\langle i_0,\ldots,i_N\rangle$}
$G = \mathit{BuildPropKReachGraph}(P,\sinit,\sfinal,\trans,K)$\\
$\pi = \mathit{GetShortestPath}(G,\sinit,\sfinal)$\\
$\chain = \mathit{GetChain}(G,\pi,\trans)$\\ 
\Return \rronly{$\chain$}
\caption{\label{alg:minchain1} 
Compute test case chain}
\end{algorithm}
%%%%%%%%%%%%%%%%%%%%%%%%%%%%%%%%%%%%%%%%%%%%%%%%%%%%%%%%%%%%%%%%%%%%%%%%%%%%%%%%
}

%+++++++++++++++++++++++++++++++++++++++++++++++++++++++++++++++++++++++++++++++
\subsection{Abstraction: Property K-Reachability Graph}
%+++++++++++++++++++++++++++++++++++++++++++++++++++++++++++++++++++++++++++++++

The \emph{property $K$-reachability graph} is an abstraction of the original
program by a weighted, directed graph $(V,E,W)$, with 
\begin{compactitem}

\item vertices $V=\Pi \cup\{\sinit,\sfinal\}$, all defining property
  assumptions, including formulas describing the sets $\sinit$ and
  $\sfinal$,

\item edges $E \subseteq E_{target} \subset V\times V$, as explained below, and 

\item an edge labelling $W: E\rightarrow \nat$ assigning to each
  $(\assume,\assume')\in E$ the minimal number of steps bounded by $K$
  needed to reach some state satisfying $\assume'$ 
%the property trigger $\assume^{'0}$ 
  from any state satisfying  $\assume$ %$\widehat{\assume}$
  according to the program's transition relation $\trans$.

\end{compactitem}

\smallskip\noindent Fig.~\ref{fig:ex-chain1} shows the property
2-reachability graph for our example.

%-------------------------------------------------------------------------------
\smallskip\paragraph{Graph construction}
\rronly{
The graph is constructed by the function
\textit{BuildProp\-KReachGraph} (Alg.~\ref{alg:buildgraph1}).
The main work is done by the function
$\getkedges$ $((V,E,W),\trans,E_{target},k)$, which computes the subset
of edges $E_k$ that have weight $k$ in the set of interesting edges
$E_{target}$.
The constructed graph contains an edge $(\assume,\assume')$ with weight
$k$ iff for the two properties with assumptions $\assume$ and $\assume'$,
a state in $\assume'$ is reachable from a state $\assume$ in $k\leq K$ steps, and $k$ is the
minimal number of steps for reaching $\assume'$ from $\assume$.
We stop the construction of the graph if a path has been found (line~5).
$\mathit{ExistsPath}$ is explained below.
If we fail to find a path before reaching a given
reachability bound $K$, or there is no path although the graph
contains all edges in $E_{target}$, then we abort (line 6).}
\pponly{ 
The graph is constructed by iteratively calling a function
  $\getkedges$ that returns the subset of edges that have weight $k$
  in the set of interesting edges 
$
E_{target} =
  \left(\bigcup_{\assume_j\in\Pi}
  \{(\sinit,\assume_j),(\assume_j,\sfinal)\}\right) \cup
  \{(\assume_j,\assume_\ell)\mid \assume_j,\assume_\ell\in\Pi, j\neq \ell\}.
$
  $E_{target}$ contains all pairwise links between the nodes
  $\assume_j$, links from $\sinit$ to all nodes $\assume_j$, and from
  every $\assume_j$ to $\sfinal$.
  $\getkedges$ (\eg~implemented using constraint solving) is called
  for increasing values of $k$ and the obtained edges are added to the
  graph until a \emph{covering path} exists, \ie, a path from $\sinit$
  to $\sfinal$ visiting all nodes at least once.
If we fail to find a path before reaching a given
reachability bound $K$, or there is no path although the graph
contains all edges in $E_{target}$, then we abort.
The constructed graph contains an edge $(\assume,\assume')$ with weight
$k$ iff for the two properties with assumptions $\assume$ and $\assume'$,
a state in $\assume'$ is reachable from a state $\assume$ in $k\leq K$ steps, and $k$ is the
minimal number of steps for reaching $\assume'$ from $\assume$.
}

%%%%%%%%%%%%%%%%%%% example chaining %%%%%%%%%%%%%%%%%%%%%%%%%%%%%%%%%%%%%%%%%%%
\begin{figure}[t]
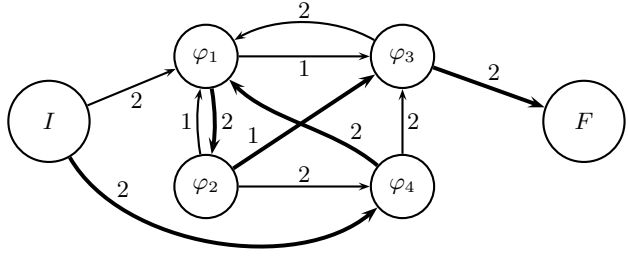

\centering
\vspace*{1ex}
\oldfootnotesize
\psset{arrows=->,labelsep=1pt}
\begin{tabular}{c@{\hspace{3em}}c@{\hspace{5em}}c@{\hspace{4em}}c}
 & \circlenode{p1}{\makebox[1.75em]{$\assume_1$}} & \circlenode{p3}{\makebox[1.75em]{$\assume_3$}} \\[-1ex]
\circlenode{i}{\makebox[2.5em]{$\sinit$}} &&& \circlenode{f}{\makebox[2.5em]{$\sfinal$}} \\[-1ex]
& \circlenode{p2}{\makebox[1.75em]{$\assume_2$}} &  \circlenode{p4}{\makebox[1.75em]{$\assume_4$}} 
\end{tabular}
\ncline{i}{p1}\nbput{2}
\nccurve[linewidth=1.5pt,angleA=-80,angleB=80]{p1}{p2}\naput{2}
\nccurve[angleA=100,angleB=-100]{p2}{p1}\naput{1}
\nccurve[angleA=0,angleB=180]{p1}{p3}\nbput{1}
\nccurve[linewidth=1.5pt,angleA=-60,angleB=-140]{i}{p4}\naput[npos=0.2]{2}
\nccurve[angleA=150,angleB=30]{p3}{p1}\nbput{2}
\nccurve[angleA=0,angleB=180]{p2}{p4}\naput{2}
\nccurve[linewidth=1.5pt,angleA=140,angleB=-45]{p4}{p1}\nbput[npos=0.2]{2}
\ncline{p4}{p3}\nbput{2}
\ncline[linewidth=1.5pt]{p3}{f}\naput{2}
\ncline[linewidth=1.5pt]{p2}{p3}\naput[npos=0.2]{1}
\vspace*{2ex}
\caption{\label{fig:ex-chain1} Test case chaining: property
  $K$-reachability graph (for $K=2$) and minimal test case chain of
  length $n=9$ (bold edges) for our example (Fig.~\ref{fig:ex-cruise}).}
\end{figure}
%%%%%%%%%%%%%%%%%%%%%%%%%%%%%%%%%%%%%%%%%%%%%%%%%%%%%%%%%%%%%%%%%%%%%%%%%%%%%%%%

\rronly{
%%%%%%%%%%%%%%%% build prop K reach graph %%%%%%%%%%%%%%%%%%%%%%%%%%%%%%%%%%
\begin{algorithm}[ht]
\KwIn{property assumptions $\Pi$, formulas $\sinit$, $\sfinal$, 
transition function $\trans$, reachability bound $K$}
\KwOut{weighted, directed graph $(V,E,W)$}
$V\gets \Pi\cup\{\sinit,\sfinal\}$\\
$E\gets\emptyset$, $W\gets\emptyset$\\
$E_{target}\gets \left(\bigcup_{\assume_j\in\Pi} \{(\sinit,\assume_j),(\assume_j,\sfinal)\}\right) \cup
\{(\assume_j,\assume_k)\mid \assume_j,\assume_k\in\Pi, j\neq k\}$\\
$k\gets 0$\\
%$\pi\gets\langle\rangle$\\
%\While{$\pi=\langle\rangle$}{
\While{$\neg\mathit{ExistsPath}((V,E,W),\sinit,\sfinal)$}{
  \lIf{$k>K \vee E_{target}=\emptyset$}{\KwAbort ``no chain found for given bound $K$''}\\
  \KwLet $E_k = \getkedges((V,E,W),\trans,E_{target},k)$\\
  $E\gets E \cup E_K$, 
  $E_{target}\gets E_{target} \setminus E_k$\\
  \lForall{$e\in E_k$}{$W\gets W\cup \{e\mapsto k\}$}\\
  $k\gets k+1$
}
\Return $(V,E,W)$
\caption{\label{alg:buildgraph1}
$\mathit{BuildPropKReachGraph}$
}
\end{algorithm}
%%%%%%%%%%%%%%%%%%%%%%%%%%%%%%%%%%%%%%%%%%%%%%%%%%%%%%%%%%%%%%%%%%%%%%%%%%%%%%%%
}

%-------------------------------------------------------------------------------
\smallskip\paragraph{Existence of a covering path}
\rronly{
 Alg.~\ref{alg:buildgraph1} requires to check for the existence of a
 covering path (function $\mathit{ExistsPath}$) in each
 iteration. 
}
The existence of a covering path can be formulated as a
reachability problem in a directed graph:
\begin{lemma}\label{lem:existspath}
Let $(V,E)$ be a directed graph of the kind described above. Then, 
there is a covering path from $\sinit$ to $\sfinal$
iff 
\begin{compactenum}[(1)]
\item all vertices are reachable from $\sinit$,
\item $\sfinal$ is reachable from all vertices, and 
\item for all pairs of vertices $(v_1,v_2) \in (V\setminus\{\sinit,\sfinal\})^2$,\\
\begin{inparaenum}[(a)]
\item $v_2$ is reachable from $v_1$ or 
\item $v_1$ is reachable from $v_2$.
\end{inparaenum}
\end{compactenum}
\end{lemma}
%% !!! REF to this proof from bullet 2 in \S\ref{sec:compl}
\begin{proof}
In the transitive closure $(V,E')$ of $(V,E)$, $v_2$ is reachable
from $v_1$ iff there exists an edge $(v_1,v_2)\in E'$.

$(\Longrightarrow)$: conditions (1) and (2) are obviously necessary.
Let us assume that we have a covering path $\pi$ and there are vertices
$(v_1,v_2)$ which neither satisfy (3a) nor (3b). Then neither
$\langle v_1,\ldots,v_2\rangle$ nor $\langle v_2,\ldots,v_1\rangle$ can be a subpath of $\pi$,
which contradicts the fact that $\pi$ is a covering path.

$(\Longleftarrow)$: Any vertex is reachable from $\sinit$ (1), so let us choose $v_1$. 
From $v_1$ we can reach another vertex $v_2$ (3a), or, at least, $v_1$ is reachable
from another vertex $v_2$ (3b), but in the latter case, since $v_2$ is
reachable from $\sinit$, we can go first to $v_2$ and then to $v_1$.
Induction step: Let us assume we have a path $\langle \sinit,v_1,\ldots,v_k \rangle$.
If there is a vertex $v'$ that is reachable from $v_k$ (3a) we add it to our current
path $\pi$. If $v'$ is unreachable from $v_k$, then by (3b), $v_k$
must be reachable from $v'$, and there is a
$v_i, i<k$ in $\pi=\langle \sinit,\ldots,v_k\rangle$ from which it is reachable and
in this case we obtain the path
$\langle \sinit,\ldots,v_i,v',v_{i+1},\ldots,v_k\rangle$;
if there is no such $v_i$ then, at last by (1), $v'$ is reachable from
$\sinit$, so we can construct the path $\langle \sinit,v',\ldots,v_k \rangle$.
$\sfinal$ is reachable from any vertex (2), thus, we can complete the
covering path as soon as all other vertices have been covered. \qed
\end{proof}
Reachability can be checked in constant time on the transitive closure
of the graph.  Hence, the overall existence check has complexity
$\mathcal{O}(|V|^3)$.

%+++++++++++++++++++++++++++++++++++++++++++++++++++++++++++++++++++++++++++++++
\subsection{Optimisation: Shortest Path Computation}
%+++++++++++++++++++++++++++++++++++++++++++++++++++++++++++++++++++++++++++++++

The next step is to compute the shortest path\rronly{ (function
$\mathit{GetShortestPath}$ in Alg.~\ref{alg:minchain1})} covering all
nodes in the property K-reachability graph.
Such a path is not necessarily Hamiltonian; revisiting nodes is
allowed. However, we can compute the transitive closure of the graph
using the Floyd-Warshall algorithm\rronly{~\cite{Flo62}} (which
preserves minimality),
%(\ie, all pairwise shortest paths, $\mathcal{O}(|V|^3)$) 
and then compute a Hamiltonian path from $\sinit$ to
$\sfinal$.
If we do not have a Hamiltonian path solver, we can add an edge from
$\sfinal$ to $\sinit$ and pass the problem to an \emph{asymmetric
  travelling salesman problem} (ATSP) solver \rronly{(referred to as
$\mathit{SolveATSP}$ in the sequel) }that gives us the shortest
circuit that visits all vertices exactly once.
We cut this circuit between $\sfinal$ and $\sinit$ to obtain the
shortest path $\pi$.

\rronly{
\begin{lemma}[Minimum covering path]\label{lem:mincoverpath}
  Let $(V,E',W')$ be the transitive closure of a
  weighted directed graph $(V,E,W)$, and $\sinit, \sfinal \in V$. Then,
  $\mathit{SolveATSP}$ $(V,E'\cup\{(\sfinal,\sinit)\},W'\cup\{(\sfinal,\sinit)\mapsto 1\})$ returns a permutation $\pi=\langle v_0,\ldots,v_{|V|-1}\rangle$ of vertices
  $V$ such that
  $\langle v_i=\sinit,\ldots,v_{|V|-1},v_0,\ldots,v_{i-1}=\sfinal \rangle$ is a
  minimum covering path from $\sinit$ to $\sfinal$.
\end{lemma}
\begin{proof}
  $(V,E,W)$ has a covering path
  $\langle\ldots,v,v',v,v'',\ldots\rangle$ that is non-Hamiltonian, then
  $(V,E',W')$ 
has a Hamiltonian path $\langle\ldots,v,v',v'',\ldots\rangle$
  because $v''$ is reachable from $v'$.
  
  Any Hamiltonian circuit $\langle v_0,\ldots,v_{|V|-1}\rangle$ returned
  by $\mathit{SolveATSP}$ must contain the edge $(v_i,v_{(i+1)\mod
    |V|})=(\sfinal,\sinit)$ because $(\sfinal,\sinit)$ is
  the only (and hence the cheapest) edge for reaching $\sinit$ from $\sfinal$.

  The obtained path has minimum length because the transitive closure
  preserves optimality ($W(v_1,v_2)+W(v_2,v_3)=W(v_1,v_3)$).\qed
 \end{proof}
}
For our example, the shortest path has length 9, given as
bold edges in Fig.~\ref{fig:ex-chain1}.

%+++++++++++++++++++++++++++++++++++++++++++++++++++++++++++++++++++++++++++++++
\subsection{Concretisation: Computing the Test Case Chain}\label{sec:concr}
%+++++++++++++++++++++++++++++++++++++++++++++++++++++++++++++++++++++++++++++++

Once we have found a minimum covering path $\pi$ in the property $K$-reachability
graph abstraction, we have to compute the inputs corresponding to it
 in the concrete program. This is done by the function 
$\chpath(\pi,\trans,W)$ which takes an abstract path 
$\pi=\langle \assume_1,\ldots,\assume_{|V|}\rangle$ and returns the input sequence
$\langle i_0,\ldots,i_n\rangle$ corresponding to a concrete path
with the reachability distances between each
$(\assume_j,\assume_{j+1})\in\pi$ 
given by the edge weights $W(\assume_j,\assume_{j+1})$.
Typically, $\chpath$ involves constraint solving;
we will discuss our implementation in \S\ref{sec:inst}.
\rronly{
Hence, $\mathit{GetChain}$ in Alg.~\ref{alg:minchain1}
corresponds to a call to $\chpath(\pi,\trans,W)$ and returning the obtained
input sequence.
}

For our example, we obtain, for instance, the sequence 
$\langle$\textit{gas, acc, button, dec, dec, gas, dec, brake, button}$\rangle$ 
corresponding to the bold edges in Fig.~\ref{fig:ex-cruise}.

%+++++++++++++++++++++++++++++++++++++++++++++++++++++++++++++++++++++++++++++++
\subsection{Optimality}
%+++++++++++++++++++++++++++++++++++++++++++++++++++++++++++++++++++++++++++++++
Since the (non-)existence or the optimality of a chain in the
$K$-reachability abstraction does not imply the (non-)existence or the
optimality of a chain in the concrete program, the success of this
procedure can only be guaranteed under certain conditions, which we
now discuss.

\begin{lemma}[Single-state property triggers]\label{lem:opt1}
  If (1) the program and the properties admit a test case chain, (2) all
  triggers \rronly{$\widehat{\assume}$ }of properties in $P$ are singleton sets,
  and (3) the test case chain \rronly{$\chain$ }computed by
  \rronly{Alg.~\ref{alg:minchain1}}\pponly{our algorithm}
  visits each property once, then the test case chain is minimal.
\end{lemma}
\rronly{\begin{proof}}
If each property is visited once, it is guaranteed that the abstract
path contains only edges that correspond to concrete paths of minimal
length, and hence the test case chain \rronly{$\chain$}is optimal for the concrete program.
Otherwise, for a subpath $(\assume,\assume',\assume,\assume'')$, there might exist an edge
$(\assume',\assume'')$ with $W(\assume',\assume'')<W(\assume',\assume)+W(\assume,\assume'')$ that
is only discovered for higher values of $K$.\rronly{ \qed \end{proof}
} %
For finite state systems, there is an upper bound for $K$, the reachability
diameter\pponly{, \ie,~the maximum (finite) length $d$ of a path in the set of shortest
paths between any pair of states $s_i,s_j \in\Sigma$}
\pponly{\cite{KS03}}\rronly{\cite{BAS02,KS03} }\pponly{. Beyond $d$}\rronly{beyond that} 
we will not discover shorter pairwise links.
\rronly{
\begin{definition}[Reachability diameter]
The reachability diameter $d$ of a system $(\Sigma,\Upsilon,
\trans,\init)$ is the maximum (finite) length of a path in the set of shortest
paths between any pair of states $s_i,s_j \in\Sigma$.
\end{definition}
}
\begin{theorem}[Minimal test case chain]\label{thm:opt2}
  Let $d$ be the reachability diameter of the program, then there is a
  $K\leq d$ such that, under the preconditions (1) and (2) of
  Lem.~\ref{lem:opt1}, the test case chain \rronly{$\chain$ }\rronly{computed by
  Alg.~\ref{alg:minchain1} }is minimal.
\end{theorem}
\rronly{
\begin{proof}
  For $K=d$, it is guaranteed that the abstract path contains only
  edges of minimal length, and hence the chain is optimal w.r.t~the
  concrete program (even if properties are revisited).
\end{proof}
}

\noindent In practice, we can stop the procedure if a chain of
acceptable length is found, \ie~we do not compute the reachability
diameter but use a user-supplied bound.

%===============================================================================
\section{Generalisations}\label{sec:drop}
%===============================================================================

We will now generalise our algorithm in three ways:
\begin{compactitem}
\item \emph{Multi-state property triggers:} Dropping the assumption
  that triggers are single-state may make the concretisation
  phase fail. Under certain restrictions, we will still 
  find a test case chain if one exists, but we lose minimality.
\item Without these restrictions, we might even lose
  completeness, \ie,~the guarantee to find a chain if one exists.
\pponly{We propose an abstraction refinement
 to \emph{ensure completeness} under these circumstances.}
\rronly{We propose two methods to \emph{ensure completeness}
  under these circumstances: (1) an abstraction refinement that
  can be used with any ATSP solver, and (2) a method based on 
  restricting the optimisation problem using path constraints that 
  requires a more general solver, \eg~an Answer Set Programming (ASP) solver.}
\item \emph{Multiple chains:} Dropping the assumption about the
  existence of a single chain raises the problem of how to generate
  multiple chains. 
\end{compactitem}
\pponly{We discuss here the first two points and refer to the extended
version~\cite{SMK13b} of this paper for the third one.}

%+++++++++++++++++++++++++++++++++++++++++++++++++++++++++++++++++++++++++++++++
\subsection{Multi-State Property Triggers}
%+++++++++++++++++++++++++++++++++++++++++++++++++++++++++++++++++++++++++++++++

In practice, many properties are multi-state, \ie~preconditions (2) of
Lem.~\ref{lem:opt1} is not met.  In this case, the abstract covering
path might be infeasible in the concrete program, and hence, the naive
concretisation of \S\ref{sec:concr} might fail. We have to extend the
concretisation step to fix such broken chains.

\begin{example}[Broken chain]
Let us consider the following broken chain in our example with the properties:\\
\begin{tabular}{@{\hspace{2em}}l@{\hspace{0.5em}}l}
$p_1:$ & $\mathbf{G}\big(mode=\mathit{OFF} \wedge \neg enable \wedge button \Rightarrow \mathbf{X}~enable\big)$ \\
$p_2:$ & $\mathbf{G}\big(mode=\mathit{ON} \wedge brake \Rightarrow \mathbf{X}(mode=\mathit{DIS})\big)$ 
\end{tabular}~\\
with $\sinit = \sfinal = \{mode=OFF \wedge speed=0 \wedge \neg enable\}$.

We obtain a shortest covering path $\langle
\sinit,\assume_1,\assume_2,\sfinal \rangle$ in the abstraction with weights
$W(\sinit,\assume_1)=0$, $W(\assume_1,\assume_2)=1$, and $W(\assume_2,\sfinal)=2$.
However, Fig.~\ref{fig:ex-cruise} tells us that the path
$\langle \sinit,\assume_1,\assume_2 \rangle$ is not feasible in a
single step, but requires two steps, as illustrated in
Fig.~\ref{fig:brokenchain1}.
\end{example}

%%%%%%%%%%%%%%%%%%%%%%%%%%%%%%%%%%%%%%%%%%%%%%%%%%%%%%%%%%%%%%%%%%%%%%%%%%%%%%%%
\begin{figure}[t]
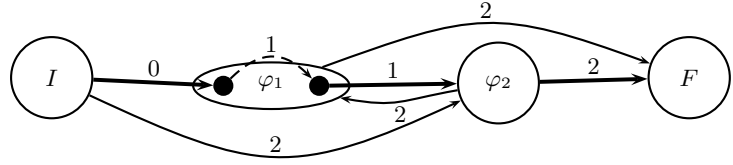

\centering
\vspace*{1ex}
\psset{arrows=->,labelsep=1pt}
\begin{tabular}{c@{\hspace{4em}}c@{\hspace{4em}}c@{\hspace{4em}}c}
\circlenode{i}{\makebox[2.5em]{$\sinit$}} & 
\ovalnode{p2}{\makebox[4em]{\circlenode[fillcolor=black,fillstyle=solid]{p2a}{}\hspace{1em}$\assume_1$\hspace{1em}\circlenode[fillcolor=black,fillstyle=solid]{p2b}{}}}
&
\circlenode{p1}{\makebox[2.5em]{$\assume_2$}} &
\circlenode{f}{\makebox[2.5em]{$\sfinal$}} \\[5ex]
\end{tabular}
\ncline[linewidth=1.5pt]{i}{p2a}\naput{$0$}
\ncline[linewidth=1.5pt]{p2b}{p1}\naput{$1$}
\nccurve[linestyle=dashed,angleA=45,angleB=135,ncurv=1]{p2a}{p2b}\naput{$1$}
\nccurve[angleA=-170,angleB=-10,ncurv=1]{p1}{p2}\naput{$2$}
\nccurve[angleA=-25,angleB=-155,ncurv=1]{i}{p1}\naput{$2$}
\nccurve[angleA=20,angleB=160,ncurv=1]{p2}{f}\naput{$2$}
\ncline[linewidth=1.5pt]{p1}{f}\naput{$2$}
\caption{\label{fig:brokenchain1} 
Broken chain: the path $\langle \sinit,\assume_1,\assume_2 \rangle$ is not
 feasible in a single step, but requires two steps.
}
\end{figure}
%%%%%%%%%%%%%%%%%%%%%%%%%%%%%%%%%%%%%%%%%%%%%%%%%%%%%%%%%%%%%%%%%%%%%%%%%%%%%%%%

A broken chain contains an infeasible subpath $\failedpath=\langle
\assume_1,\ldots,\assume_k\rangle$ of the abstract path $\pi$ that involves
at least three vertices, such as $\langle \sinit,\assume_1,\assume_2
\rangle$ in our example above.
We extend the concretisation step \rronly{($\mathit{GetChain}$)} with
a chain repair capability.  The function $\mathit{RepairPath}$\rronly{
  as shown in Alg.~\ref{alg:getchain2}} iteratively repairs broken
chains by incrementing the weights associated with the edges of
$\failedpath$ and checking feasibility of this ``stretched'' path.
We give more details about our implementation in \S\ref{sec:inst}.

\rronly{
%%%%%%%%%%%%% get test case chain %%%%%%%%%%%%%%%%%%%%%%%%%%%%%%%%%%%%%%%%%%%%%%
\begin{algorithm}[h]
\KwIn{weighted, directed graph $(V,E,W)$, path
  $\pi$, transition relation $\trans$}
\KwOut{test case chain $\chain=\langle i_0,\ldots,i_N\rangle$}
$(\mathit{feasible},\chain,\failedpath) \gets \chpath(\pi,\trans,W)$\\
\lIf{$\mathit{feasible}$}{\Return \rronly{$\chain$}}\\
\Else{
  $(\mathit{succeeded},W,\_) \gets
  \mathit{RepairPath}(\failedpath,\trans,W)$\\
  \lIf{$\neg\mathit{succeeded}$}{ \KwAbort ``no chain found for given bound $K$''}\\
  $(\_,\chain,\_) \gets \chpath(\pi,\trans,W)$\\
  \Return \rronly{$\chain$}
}
\caption{\label{alg:getchain2} 
$\mathit{GetChain}$ with chain repair}
\end{algorithm}
%%%%%%%%%%%%%%%%%%%%%%%%%%%%%%%%%%%%%%%%%%%%%%%%%%%%%%%%%%%%%%%%%%%%%%%%%%%%%%%%
}

\begin{example}[Repaired chain]
For the broken chain in our previous example, we will 
check whether $\langle \sinit,\assume_1,\assume_2\rangle$ is feasible 
with $W(\assume_1,\assume_2)$ incremented by one.
This makes the path feasible and we obtain the chain 
$\rronly{\chain=}\langle$\textit{button,gas,brake,button}$\rangle$.
\end{example}

%-------------------------------------------------------------------------------
\bigskip\paragraph{Completeness}
The chain repair succeeds
if the given path $\pi$ admits a chain in the concrete program.
In particular, this holds when the states in each property
trigger are strongly connected:
\begin{theorem}[Multi-state strongly connected property]\label{thm:minchain2}
  If for each property trigger \rronly{$\widehat{\assume}$ }the states are strongly
  connected and there exists a test case chain then \rronly{Alg.~\ref{alg:minchain1}
  (with Alg.~\ref{alg:getchain2})}\pponly{our algorithm with chain
  repair} will find it.
\end{theorem}

In practice, many reactive systems are, apart from an initialisation
phase, strongly connected---but, as stressed above, the test case chain might not be minimal.

%+++++++++++++++++++++++++++++++++++++++++++++++++++++++++++++++++++++++++++++++
\subsection{Ensuring Completeness}\label{sec:compl}
%+++++++++++++++++++++++++++++++++++++++++++++++++++++++++++++++++++++++++++++++

If the shortest path \rronly{$\pi$ }in the abstraction does not admit a chain in
the concrete program, \pponly{our algorithm}\rronly{Alg.~\ref{alg:minchain1}} with chain repair \rronly{(Alg.~\ref{alg:getchain2}) }will fail to find a test case chain even though
one exists, \ie,~it is not complete.

\begin{example}[Chain repair fails]
  In Fig.~\ref{fig:brokenchain1}, we have found the shortest abstract
  path $\langle \sinit,\assume_1,\assume_2,\sfinal\rangle$.  Now
  assume that the right state in $\assume_1$ is not reachable from the
  left state. Then the chain repair fails. In this case, there might still be a
  (non-)minimal path in the abstraction that admits a chain: in our
  example in Fig.~\ref{fig:brokenchain1}, assuming that the left state
  in $\assume_1$ is reachable from $\sinit$ via $\assume_2$ and
  $\sfinal$ is reachable from the left state in $\assume_1$, we have
  the feasible path $\langle
  \sinit,\assume_2,\assume_1,\sfinal\rangle$.
\end{example}

%%%%%%%%%%%%%%%%%%%%%%%%%%%%%%%%%%%%%%%%%%%%%%%%%%%%%%%%%%%%%%%%%%%%%%%%%%%%%%%%
\begin{figure}[t]
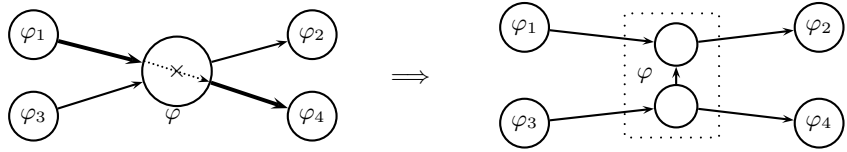

\centering
\vspace*{-1ex}
\oldfootnotesize
\psset{arrows=->}
\parbox{0.35\textwidth}{
\begin{tabular}{c@{\hspace{3em}}c@{\hspace{3em}}c}
 \circlenode{i1}{$\assume_1$} &&  \circlenode{o1}{$\assume_2$} \\[-2ex]
& \pnode{ni1}\hspace*{2.6em} \\[-3.3ex]
& \circlenode{n}{\makebox[2em]{$\times$}} & \\[-4.5ex]
& \hspace*{3.0em}\pnode{no2} \\[0ex]
 \circlenode{i2}{$\assume_3$} &$\assume$&  \circlenode{o2}{$\assume_4$}
\end{tabular}
\ncline[linewidth=1.5pt]{i1}{n}\ncline{i2}{n}\ncline{n}{o1}\ncline[linewidth=1.5pt]{n}{o2}
\ncline[linestyle=dotted,dotsep=1pt]{ni1}{no2}
}
\hspace{2em}
$\Longrightarrow$
\hspace{2em}
\parbox{0.35\textwidth}{
\begin{tabular}{c@{\hspace{3em}}c@{\hspace{3em}}c}
 \circlenode{i1}{$\assume_1$} &&  \circlenode{o1}{$\assume_2$} \\[-4ex]
& \psframebox[linestyle=dotted]{
$\begin{array}{c}$\circlenode{n1}{\makebox[1em]{}}$\\[3ex]$\circlenode{n2}{\makebox[1em]{}}$\end{array}$
} \\[-4ex]
 \circlenode{i2}{$\assume_3$} &&  \circlenode{o2}{$\assume_4$}
\end{tabular}
\ncline{i1}{n1}\ncline{i2}{n2}\ncline{n1}{o1}\ncline{n2}{o2}\ncline{n2}{n1}\naput[labelsep=8pt]{$\assume$}
}
\caption{\label{fig:refine1} Abstraction refinement
  for a failed path $\langle \assume_1,\assume,\assume_4 \rangle$ (bold arrows).
}
\end{figure}
%%%%%%%%%%%%%%%%%%%%%%%%%%%%%%%%%%%%%%%%%%%%%%%%%%%%%%%%%%%%%%%%%%%%%%%%%%%%%%%%

%%%%%%%%%%%%%%%%%%%%%%%%%%%%%%%%%%%%%%%%%%%%%%%%%%%%%%%%%%%%%%%%%%%%%%%%%%%%%%%%
\begin{figure}[t]
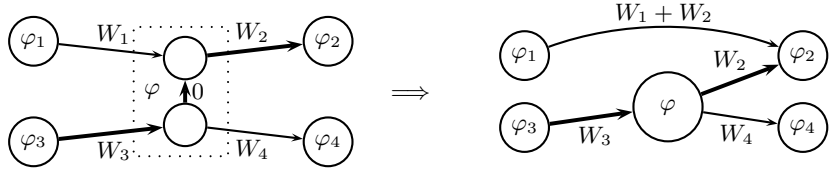

\centering
\oldfootnotesize
\psset{arrows=->,labelsep=1pt}
\parbox{0.35\textwidth}{
\begin{tabular}{c@{\hspace{3em}}c@{\hspace{3em}}c}
 \circlenode{i1}{$\assume_1$} &&  \circlenode{o1}{$\assume_2$} \\[-4ex]
& \psframebox[linestyle=dotted]{
$\begin{array}{c}$\circlenode{n1}{\makebox[1em]{}}$\\[3.5ex]$\circlenode{n2}{\makebox[1em]{}}$\end{array}$
} \\[-4ex]
 \circlenode{i2}{$\assume_3$} &&  \circlenode{o2}{$\assume_4$}
\end{tabular}
\ncline{i1}{n1}\naput{$W_1$}
\ncline[linewidth=1.5pt]{i2}{n2}\nbput{$W_3$}
\ncline[linewidth=1.5pt]{n1}{o1}\naput{$W_2$}
\ncline{n2}{o2}\nbput{$W_4$}
\ncline[linewidth=1.5pt]{n2}{n1}\naput[labelsep=8pt]{$\assume$}\nbput{$0$}
}
\hspace{2em}
$\Longrightarrow$
\hspace{2em}
\parbox{0.35\textwidth}{
\begin{tabular}{c@{\hspace{3em}}c@{\hspace{3em}}c}
 \circlenode{i1}{$\assume_1$} &&  \circlenode{o1}{$\assume_2$} \\[-1ex]
&\circlenode{n}{\makebox[2em]{$\assume$}} \\[-4ex]
 \circlenode{i2}{$\assume_3$} &&  \circlenode{o2}{$\assume_4$}
\end{tabular}
\nccurve[angleA=20,angleB=160]{i1}{o1}\naput{$W_1+W_2$}
\ncline[linewidth=1.5pt]{i2}{n}\nbput{$W_3$}
\ncline[linewidth=1.5pt]{n}{o1}\naput{$W_2$}
\ncline{n}{o2}\nbput{$W_4$}
}
\caption{\label{fig:refine2} Collapsing the property refinement group
  (box) in the refined abstraction
  to a TSP problem w.r.t. a solution path (bold arrows).
}
\end{figure}
%%%%%%%%%%%%%%%%%%%%%%%%%%%%%%%%%%%%%%%%%%%%%%%%%%%%%%%%%%%%%%%%%%%%%%%%%%%%%%%%

\rronly{
%%%%%%%%%%%%% get test case chain with abstraction refinement %%%%%%%%%%%%%%%%%%%%%%%%%%%%
\begin{algorithm}[h]
\KwIn{weighted, directed graph $(V,E,W)$, path
  $\pi$, transition relation $\trans$}
\KwOut{test case chain $\chain=\langle i_0,\ldots,i_N\rangle$}
$G \gets \{\{v\}\mid v\in V\}$ //property refinement groups\\
\While{$\true$}{
$(\mathit{feasible},\chain,\failedpath) \gets \chpath(\pi,\trans,W)$\\
\lIf{$\mathit{feasible}$}{\Return $\chain$}\\
$(\mathit{succeeded},W',\failedpath) \gets \mathit{RepairPath}(\failedpath,\trans,W)$\\
\If{$\mathit{succeeded}$}{
$(\_,\chain,\_) \gets \chpath(\pi,\trans,W')$\\
\Return $\chain$
}
$(\assume',\assume,\assume'') = \failedpath$\\
$V\gets V\cup\{\mathit{v_{new}}\}$\\
$\mathit{getGroup}(G,\assume)\gets \mathit{getGroup}(G,\assume) \cup\{\mathit{v_{new}}\}$\\
$E\gets E\cup\{(\assume',\mathit{v_{new}})\}$\\
$W (\assume',\mathit{v_{new}}) \gets W (\assume',\assume'')$\\
$E\gets E\setminus\{(\assume',\assume)\}$\\
$E\gets E\cup\{(\mathit{v_{new}},v)\mid (\assume,v)\in E\setminus\{\assume,\assume'')\}\}$\\

$\pi \gets GetCoveringPath(V,E,G)$\\
\lIf{$\pi=\langle\rangle$}{\KwAbort ``no chain found for given bound $K$''}

\ForEach{$\bar{v} \in \pi$}{
  \ForEach{$v \in \mathit{getGroup}(G,\bar{v})$}{
    \If{$v\neq \bar{v}$}{
      $E' = \{(v',v'')\mid (v',v) \in E \wedge
      (v,v'') \in E \wedge (v',v'')\notin E\}$\\
      \ForEach{$(v',v'') \in E'$}{
        $E \gets E \cup \{(v',v'') \}$\\
        $W(v',v'') \gets W(v',v) +W (v,v'')$\\
        $E \gets E \setminus E'$\\
        $V \gets V \setminus \{v\}$
      }
    } 
  }
}
$\pi \gets GetShortestPath(V,E,W)$\\
}
\caption{\label{alg:getchain3} 
$\mathit{GetChain}$ with abstraction refinement}
\end{algorithm}
%%%%%%%%%%%%%%%%%%%%%%%%%%%%%%%%%%%%%%%%%%%%%%%%%%%%%%%%%%%%%%%%%%%%%%%%%%%%%%%%
}
\rronly{
%%%%%%%%%%%%% get path %%%%%%%%%%%%%%%%%%%%%%%%%%%%%%%%%%%%%%%%%%%%%%
\begin{algorithm}[h]
\KwIn{transitive closure of directed graph $(V, E)$, property
  refinement groups $G$}
\KwOut{covering path $\pi$}
$v \gets \mathit{chooseFrom}(V)$; $V \gets V \setminus\mathit{getGroup}(G,v)$;
$\pi\gets \langle v \rangle$\\
\While{$V\neq\emptyset$}{
  $v \gets \mathit{chooseFrom}(V)$; $V \gets V \setminus
 \mathit{getGroup}(G,v)$; $v' \gets \mathit{lastElement}(\pi)$ \\
  \lIf{$(v',v)\in E$}{$\pi \gets \mathit{append}(\pi,v)$}\\
  \ElseIf{$(v,v')\in E$}{
     \lWhile{$(v',v)\notin E$}{ $v' \gets \mathit{previousElement}(\pi,v')$ } \\
     $\pi \gets \mathit{insertAfter}(\pi,v,v')$
   }
   \lElse{\Return $\langle\rangle$ //no path found}
}
\Return $\pi$
\caption{\label{alg:getcoverpath} 
$\mathit{GetCoveringPath}$}
\end{algorithm}
%%%%%%%%%%%%%%%%%%%%%%%%%%%%%%%%%%%%%%%%%%%%%%%%%%%%%%%%%%%%%%%%%%%%%%%%%%%%%%%%
}

\rronly{\paragraph{Abstraction refinement} }
To obtain completeness in this situation, 
we propose the following abstraction refinement method\rronly{ shown in Alg.~\ref{alg:getchain3}}.
\pponly{Suppose a covering path $\pi$ in the abstraction turns out to be
infeasible in the concrete program, with $\failedpath=\langle
\assume_1,\ldots,\assume_N\rangle$. }
\rronly{Suppose the chain repair of a covering path $\pi$ failed with 
$\failedpath=\langle \assume_1,\assume,\assume_4\rangle$
($\mathit{succeeded}=\false$ in line 5).}
\begin{compactenum}
\item We refine \pponly{failed vertices $\assume_2,\ldots,\assume_{N-1}$ in $\failedpath$
  by splitting them}\rronly{the graph by splitting vertex $\assume$
  in $\failedpath$} as illustrated in Fig.~\ref{fig:refine1} that
  rules out the infeasible subpath, as
  typically done by abstract refinement algorithms (lines 10--15). We call the
  vertices obtained from such splittings that belong to the same
  property a \emph{property refinement group}\rronly{ (subsets of $G$;
    the function $\mathit{getGroup}(G,v)$ returns the subset
    containing $v$)}.
\item The second part of the proof of Lem.~\ref{lem:existspath} gives
  us an $\mathcal{O}(n^2)$ algorithm \textit{GetCover\-ingPath} for finding a (non-minimal)
  covering path from $\sinit$ to $\sfinal$ in the transitive
  closure of a directed graph\rronly{ (see Alg.~\ref{alg:getcoverpath})},
  taking into account that a covering path needs to cover only one
  vertex for each property refinement group\rronly{ (called in line 16 of Alg.~\ref{alg:getchain3})}.
\item A solution $\pi$ obtained that way might be far from optimal, so we 
  exploit the TSP solver to give us a better solution $\pi'$.
  However, the refined graph does not encode the desired TSP problem
  because it is sufficient to cover only one vertex for each property
  refinement group.  Hence, given a path $\pi$, we transform the graph 
  by collapsing each
  property refinement group with respect to $\pi$ as illustrated by
  Fig.~\ref{fig:refine2}\rronly{ (lines 18--26 of
    Alg.~\ref{alg:getchain3})}.  
The obtained graph is handed over to the
  TSP solver (line 27).  Note that the transformations do
  not preserve optimality, because, \eg~in Fig.~\ref{fig:refine2},
  the edge $(\assume_1,\assume_2)$ would cover $\assume$ in a concrete path but not in
  the transformed, refined abstract graph.
\item We try to compute a concrete test case chain for the covering
  path\rronly{ (lines 3--8)}. If this fails, we iterate the
  refinement process.
\end{compactenum}

In each iteration\rronly{ (line 2)} of the abstraction refinement
algorithm, a node in the graph is split such that a concrete spurious
transition is removed from the abstraction, \ie~the transition system
structure of the program inside the property assumptions is made
explicit in the abstraction.  Provided the existence of a test case chain,
since there is only a finite number of transitions, the abstraction
refinement will eventually terminate, and a covering path will be
found that can be concretised to a test case chain.

\begin{example}[Abstraction refinement]
Assume, as in the previous example, that the right state in $\assume_1$ in
Fig.~\ref{fig:brokenchain1} is not reachable from the left state. Then
the abstraction refinement will split $\assume_1$ into two vertices.
Suppose that \textit{GetCover\-ingPath}\rronly{ (Alg.~\ref{alg:getcoverpath})} returns the covering path
$\pi=\langle\sinit,\assume_2,\assume_1,\assume_2,\sfinal\rangle$.\footnote{It
will actually return the better result for this particular example.} Then collapsing the two nodes
belonging to $\assume_1$ w.r.t. $\pi$ will remove the edge from $\sinit$
to $\assume_1$.  The TSP solver will optimise $\pi$ and find the shorter
path $\langle\sinit,\assume_2,\assume_1,\sfinal\rangle$.
\end{example}

\rronly{
\paragraph{Path constraints.}
The fundamental problem about a failed path is that it represents information
about at least two edges that we cannot encode as an equivalent TSP.
We would need a TSP solver that can deal with side conditions like the following: the
solution must not contain vertices $v_1,v_2,v_3$ in this particular order
for any infeasible subpath $\langle v_1,v_2,v_3\rangle$ in $\failedpath$.
Similar difficulties arise concerning \emph{minimality}: 
here, we would have to add ``path weights'' that penalise a solution 
if it contains a certain path.
Since our experimental results (\S\ref{sec:exp}) suggest that the bottleneck of the
approach lies rather in solving reachability queries than TSPs, we
can opt for using answer set programming (ASP) solvers (\eg~\cite{GKK+11}), which
are far less efficient in solving TSPs, but they allow us to specify
arbitrary side conditions.

\begin{example}[Path constraints]
Consider the graph in Fig.~\ref{fig:brokenchain1}. We can encode the
TSP problem in ASP as follows (cf.~\cite{GKK+11}):\\[1ex]
{\oldfootnotesize
\hspace*{1em}\texttt{V(I,phi1,phi2,F).}\\[1ex]
\hspace*{1em}\texttt{E(I,phi1).    weight(I,phi1,0).}\\
\hspace*{1em}\texttt{E(I,phi2).    weight(I,phi2,2).}\\
\hspace*{1em}\texttt{E(phi1,phi2). weight(phi1,phi2,1)}.\\
\hspace*{1em}\texttt{E(phi1,F).    weight(phi1,F,2).}\\
\hspace*{1em}\texttt{E(phi2,phi1). weight(phi2,phi1,2).}\\
\hspace*{1em}\texttt{E(phi2,F).    weight(phi2,F,2).}\\[1ex]
\hspace*{1em}\texttt{\{ cycle(X,Y) : E(X,Y) \} I :- V(X).}\\
\hspace*{1em}\texttt{\{ cycle(X,Y) : E(X,Y) \} I :- V(Y).}\\[1ex]
\hspace*{1em}\texttt{ reached(Y) :- cycle(I,Y).}\\
\hspace*{1em}\texttt{ reached(Y) :- cycle(X,Y), reached(X).}\\
\hspace*{1em}\texttt{:- V(Y), not reached(Y).}\\[1ex]
\hspace*{1em}\texttt{\#minimize [ cycle(X,Y) : weight(X,Y,C) = C
  ].}\\[1ex]
}
Assume, again, that the right state in $\assume_1$ in
Fig.~\ref{fig:brokenchain1} is not reachable from the left state
so that we obtain $\failedpath=\langle
\sinit,\assume_1,\assume_2\rangle$.
Then we can exclude $\failedpath$ by adding\\[1ex]
{\oldfootnotesize
\hspace*{1em}\texttt{twopath(X,Y,Z) :- cycle(X,Y), cycle(Y,Z).}\\
\hspace*{1em}\texttt{-twopath(I,phi1,phi2).}\\[1ex]
}
to the ASP problem.
The ASP solver will return the shortest covering path that
does not contain $\failedpath$, 
\ie~$\langle\sinit,\assume_2,\assume_1,\sfinal\rangle$.
\end{example}
}

\rronly{
%+++++++++++++++++++++++++++++++++++++++++++++++++++++++++++++++++++++++++++++++
\subsection{Multiple Chains}
%+++++++++++++++++++++++++++++++++++++++++++++++++++++++++++++++++++++++++++++++

%%%%%%%%%%%%% partitioning in multiple chains %%%%%%%%%%%%%%%%%%%%%%%%%%%%%%%%%%
\begin{algorithm}[b!]
\KwIn{directed graph $(V,E)$}
\KwOut{partition $S$ of $V$}
$R =$ set of pairs $(v_i,v_j)\in V$ that do not satisfy 
condition (3) of Lem.~\ref{lem:existspath}.\\
$S\gets\emptyset$, $Q\gets V$ \\
\Forall{$(v_i,v_j)\in R$}{
  $Q \gets Q \setminus \{v_i,v_j\}$\\
  \lIf{$S=\emptyset$}{$S\gets \{(\{v_i\},\{v_j\}),(\{v_j\},\{v_i\})\}$}\\
  \Else{
    \Forall{$P=(P^+,P^-)\in S$}{
    \lIf{$v_i\in P^+ \wedge v_j\in P^+$}{$S\gets S\setminus P$}\\
    \lElseIf{$v_i\in P^- \wedge v_j\notin P^- \wedge v_j\notin
      P^+$}{$P \gets (P^+\cup \{v_j\},P^-)$}\\
    \lElseIf{$v_i\notin P^- \wedge v_j\notin P^- \wedge v_j\in
      P^+$}{$P \gets (P^+,P^-\cup \{v_i\})$}\\
    \lElseIf{$v_j\in P^- \wedge v_i\notin P^- \wedge v_i\notin
      P^+$}{$P \gets (P^+\cup \{v_i\},P^-)$}\\
    \lElseIf{$v_j\notin P^- \wedge v_i\notin P^- \wedge v_i\in
      P^+$}{$P \gets (P^+,P^-\cup \{v_j\})$}\\
    \lElseIf{$v_i\notin (P^+\cup P^-) \wedge v_j\notin (P^+\cup P^-)$}{
      $S \gets S \cup (P^+ \cup \{v_i\}, P^- \cup \{v_j\})$;
      $P \gets (P^+ \cup \{v_j\}, P^- \cup \{v_i\})$}
  }
}
}
$S^+ = \mathit{MinCover}(V,\{P^+\mid (P^+,\_)\in S\})$\\
\KwChoose $P^+\in S^+$: $P^+\gets P^+\cup Q$\\
\lForall{$P^+\in S^+$}{$P^+\gets P^+\cup \{\sinit,\sfinal\}$}\\
\Return $S^+$
\caption{\label{alg:partition1} 
Multiple chains: Partitioning the vertex (property) set such that each partition
element admits a single chain}
\end{algorithm}
%%%%%%%%%%%%%%%%%%%%%%%%%%%%%%%%%%%%%%%%%%%%%%%%%%%%%%%%%%%%%%%%%%%%%%%%%%%%%%%%

We can relax our problem to systems that do not admit single chains. 
Those systems still have to satisfy conditions (1) and (2) of
Lem.~\ref{lem:existspath} in order to guarantee the existence of
multiple covering chains.

We can detect that a system does not admit a single chain if
\begin{compactitem}
\item the $N$-reachability property graph has no chain (where $N$ is
  the reachability diameter of the system), or 
\item the chain repair or abstraction refinement process fails.
\end{compactitem}
% or when the state invariant is not strong enough!

We use Lem.~\ref{lem:existspath} to devise an
algorithm for computing a partition $\{P_1,\ldots,P_n\}$ of $P$ (see Alg.~\ref{alg:partition1}) and apply
Alg.~\ref{alg:minchain1} for each $P_i$.
If the chain repairing fails for a $P_i$, we compute a partition for the 
refined property graph. Finding the smallest partition is equivalent
to the problem of finding a vertex colouring with minimal chromatic
number (NP-hard).
In Alg.~\ref{alg:partition1}, the set $S$ contains pairs of sets
$(P^+,P^-)$. $P^+$ contains the vertices that will form an equivalence
class. $P^-$ keeps track of
the vertices that are not allowed to be added to $P^+$.
Lines 3 to 13 compute all subsets of $V$ that are consistent with
condition (3) of Lem.~\ref{lem:existspath} ($F$).
Line 14 removes the redundant subsets (minimal set cover) and,
finally, in line 15 and 16, the remaining vertices $Q$ are added to some 
element of the partition, and $\sinit$ and $\sfinal$ are added to all partitions.
}

%===============================================================================
\section{Test-Case Generation with Bounded Model Checking}\label{sec:inst}
%===============================================================================

\rronly{
%%%%%%%%%%%%% check path %%%%%%%%%%%%%%%%%%%%%%%%%%%%%%%%%%%%%%%%%%%%%%
\begin{algorithm}[t]
\KwIn{path $\pi$, transition relation $\trans$, weights W}
\KwOut{whether $\pi$ is $\mathit{feasible}$, $\mathit{inputs}$
  associated to $\pi$ if feasible, $\failedpath\subseteq\pi$ if
  infeasible}
     $\mathit{inputs} \gets \langle\rangle$\\
     $\failedpath \gets \langle\rangle$\\
     $(\mathit{feasible},\mathit{assignment},\mathit{unsat\_core}) = 
          SAT(\mathit{BuildPath}(\pi,\trans,W))$\\
     \If{$\mathit{feasible}$}{
         \KwLet $(s_0,i_0,s_1,i_1,\ldots,s_K,i_K)=assignment$\\
         $\mathit{inputs} \gets \langle i_0,\ldots,i_N\rangle$
      }
      \Else{
         $\failedpath\gets \mathit{getFailedPath}(\mathit{unsat\_core},\pi)$
      }
\Return $(\mathit{feasible},\mathit{inputs},\failedpath)$
\caption{\label{alg:checkpath} 
$\chpath$}
\end{algorithm}
%%%%%%%%%%%%%%%%%%%%%%%%%%%%%%%%%%%%%%%%%%%%%%%%%%%%%%%%%%%%%%%%%%%%%%%%%%%%%%%%

%%%%%%%%%%%%% build path formula %%%%%%%%%%%%%%%%%%%%%%%%%%%%%%%%%%%%%%%%%%%%%%
\begin{algorithm}[t]
\KwIn{path $\pi$, transition relation $\trans$, weights $W$}
\KwOut{path formula $\Phi$}
\Return $\mathit{BuildPathRec}(\pi,0,\true)$\\
\Func{$\mathit{BuildPathRec}(\pi,k,\Phi)$}{
\If{$\pi=\langle (\varphi,\_)\rangle$}{
   \Return $\Phi \wedge \varphi(s_k)$
}
\Else{
  \KwLet $(v,\pi_{tail}) = \pi$\\
  \KwLet $(v',\_) = \pi_{tail}$\\
  \KwLet $k_{end}= k+W(v,v')$\\
  \KwLet $(\varphi,\psi) = v$\\
  \Return $\Phi \wedge \varphi(s_k,i_k) \wedge \psi(s_{k+1}) \wedge \bigwedge_{k+1\leq j\leq
    k_{end}}T(s_{j-1},s_j) \wedge \mathit{BuildPathRec}(\pi_{tail},k_{end},\Phi)$
}
}
\caption{\label{alg:buildpath} 
$\mathit{BuildPath}$}
\end{algorithm}
%%%%%%%%%%%%%%%%%%%%%%%%%%%%%%%%%%%%%%%%%%%%%%%%%%%%%%%%%%%%%%%%%%%%%%%%%%%%%%%%
}

\rronly{
%%%%%%%%%%%%%% get k reach edges %%%%%%%%%%%%%%%%%%%%%%%%%%%%%%%%%%%%%%%%%%%%%%%
\begin{algorithm}[t]
\KwIn{weighted, directed graph $(V, E,W)$, transition relation $\trans$,  edges to be considered $E_S$, number of steps $K$}
\KwOut{$K$-reach edges $E_K\subseteq E_S$}
$\mathit{from\_to}\gets E_S$\\
$E_K\gets \emptyset$\\
$(sat,assignment) \gets \chkreach(\mathit{from\_to},\trans,K)$\\
\While{$sat$}{
    \KwLet $(s_0,i_0,s_1,i_1,\ldots,s_K,i_K)=assignment$\\
    \Forall{$v,v' \in V:$ $(\varphi,\psi)=v,(\varphi',\_)=v': \varphi(s_0,i_0)
      \wedge \psi(s_1) \wedge \varphi'(s_K)$}{
      $E_K\gets E_K \cup \{(v,v')\}$\\
      $\mathit{from\_to}\gets \mathit{from\_to} \setminus \{(v,v')\}$\\
    }
    $(sat,assignment,\_) \gets \chkreach(\mathit{from\_to},\trans,K)$\\
}
\Return $E_K$
\caption{\label{alg:kreach} 
$\getkedges$}
\end{algorithm}
%%%%%%%%%%%%%%%%%%%%%%%%%%%%%%%%%%%%%%%%%%%%%%%%%%%%%%%%%%%%%%%%%%%%%%%%%%%%%%%%

%%%%%%%%%%%%% repair path %%%%%%%%%%%%%%%%%%%%%%%%%%%%%%%%%%%%%%%%%%%%%%
\begin{algorithm}[t]
\KwIn{$\failedpath$, transition relation $\trans$, weights $W$, 
reachability bound $K$}
\KwOut{updated weights $W$}
$\sigma \gets \mathit{FirstElement}(\failedpath)$\\
\Forall{$e=(\assume_j,\assume_{j+1})\in \failedpath$}{
$\mathit{feasible} \gets \false$ \\
\While{$\neg \mathit{feasible}$}{
  $(\mathit{sat},\mathit{assignment},\_) \gets \chpath(\langle
 \sigma,\assume_{j+1}\rangle,\trans,W)$ \\
  \lIf{$\neg \mathit{feasible}$}{$W(e) \gets W(e)+1$} \\
  \Else{
    \KwLet $\langle s_0,\ldots\rangle = \mathit{assignment}$ \\
    $\sigma \gets s_0$}
  \lIf{$W(e)>K$} \Return $(\false,W,\langle \assume_{j-1},\assume_j,\assume_{j+1}\rangle)$
}
}
\Return $(\true,W,\langle\rangle)$
\caption{\label{alg:repairpath1} 
$\mathit{RepairPath}$ by concrete chaining}
\end{algorithm}
%%%%%%%%%%%%%%%%%%%%%%%%%%%%%%%%%%%%%%%%%%%%%%%%%%%%%%%%%%%%%%%%%%%%%%%%%%%%%%%%
}

The previous sections abstract from the
actual backend implementation of the functions 
$\getkedges$, $\chpath$, and $\mathit{RepairPath}$.
In this work, we use bounded model checking to provide an efficient implementation.
Alternative instantiations could be based on symbolic execution, for example.

%-------------------------------------------------------------------------------
\paragraph{BMC-based test case generation}
Bounded model checking (BMC) \cite{CBRZ01} can be used to check the
existence of a path $\pi=\langle s_0,s_1,\ldots,s_K\rangle$ of
increasing length $K$ from $\phi$ to $\phi'$.
This check is performed by deciding satisfiability of the following formula
using a SAT solver:
\negskip
\begin{equation}\label{equ:query}
\phi(s_0)\wedge\bigwedge_{1\leq k\leq K} \trans(s_{k-1},i_{k-1},s_k) \wedge \phi'(s_K)
\end{equation}
\negskip

If the SAT solver returns the answer \emph{satisfiable}, it also provides a
satisfying assignment $(s_0,i_0,s_1,i_1,\ldots,s_{K-1},i_{K-1},s_K)$.
The satisfying assignment represents one possible path $\pi=\langle
s_0,s_1,\ldots,s_K\rangle$  from $\phi$
to $\phi'$  and identifies the corresponding input sequence $\langle i_0,\ldots,i_{K-1}\rangle$.
Hence, a test case $\langle i_0,\ldots,i_{K-1}\rangle$ covering
a property with assumption $\assume(s,i)$ can be generated by checking
satisfiability of a path from $\sinit$ to $\assume$.

%-------------------------------------------------------------------------------
\paragraph{Instantiation}
\rronly{For implementing Alg.~\ref{alg:minchain1} with chain repair
(Alg.~\ref{alg:getchain2}) we have to provide the functions
$\chpath$, $\getkedges$, and $\mathit{RepairPath}$.

We consider a SAT solver to be a function $SAT: \phi \mapsto
(\mathit{sat},\mathit{assignment},$ $\mathit{unsat\_core})$ where
$\mathit{assignment}$ contains a satisfying assignment if $\phi$ is
$sat$ and otherwise $\mathit{unsat\_core}$ is a minimal formula such that
$\phi \Rightarrow \mathit{unsat\_core} $ and
$\neg\mathit{unsat\_core}$ $\Rightarrow \neg\phi$.\footnote{There are alternatives to
unsatisfiability cores, \eg, the final conflict feature
of \textsc{Minisat}~\cite{EMA10}.}

Then $\chpath$ is defined as in Alg.~\ref{alg:checkpath} where 
$\mathit{BuildPath}$ (Alg.~\ref{alg:buildpath}) constructs the BMC
formula for a given path, 
and $\mathit{getFailedPath}$ converts an $\mathit{unsat\_core}$ into 
a path (which is SAT solver-specific).
}
\pponly{
$\chpath$ corresponds to a SAT query like Eq.~(\ref{equ:query}) with
$\phi=\sinit$, $\phi'=\sfinal$, conjoined with the formulas for the property
assumptions $\assume_j$ according to the covering path $\pi$.
The formula for $\trans$ includes any assumptions 
restricting the input domain.
If the query is unsatisfiable, the SAT solver can be requested to produce a
reason for the failure (\eg~in the form of an unsatisfiable core) from which
we can extract $\failedpath$.
}

\rronly{
$\getkedges$ is given as Alg.~\ref{alg:kreach}, where the function
$\chkreach(\pi,\trans,K)$ that is used for enumerating
$K$-reachability edges is implemented by checking
satisfiability of the following formula: }
\pponly{
The function $\getkedges$ that returns the 
$K$-reachability edges in a set $E_{target}$ 
is implemented by satisfiability checking the formula: 
}
\negskip
\begin{equation}\label{equ:kreach}
\left(\bigvee_{(\varphi,\varphi')\in
      E_{target}} \varphi(s_0,i_0) \wedge
    \varphi'(s_K)\right)
   \wedge\bigwedge_{1\leq k\leq K} \trans(s_{k-1},i_{k-1},s_k) 
\end{equation}
\negskip

We iteratively check this formula using incremental SAT solving, 
``removing'' the respective terms from
the formula each time a solution satisfies $(\assume,\assume')$,
until the formula becomes unsatisfiable.
In addition to assumptions on the inputs, 
$\trans$ must also contain a state invariant, obtained, \eg~with a static
analyser.  This is necessary because, otherwise, the state satisfying
$\assume$ in Eq.~\ref{equ:kreach} might be unreachable from an
initial state.

For the chain repair $\mathit{RepairPath}$, the most efficient method
that we tested was to sequentially find a feasible weight for each of
the edges in $\failedpath$, starting the check for an edge
$(\assume_j,\assume_{j+1})$ from a concrete state in $\assume_j$ obtained
from the successful check of the previous edge
$(\assume_{j-1},\assume_j)$. \rronly{This algorithm is listed in Alg.~\ref{alg:repairpath1}.}

%===============================================================================
\section{Experimental Evaluation}\label{sec:exp}
%===============================================================================

%-------------------------------------------------------------------------------
\paragraph{Implementation}
For our experiments we have set up a tool chain (Fig.~\ref{fig:tool1})
that generates C code from \textsc{Simulink} models using
the
\textsc{Gene-Auto}\footnote{\url{http://geneauto.gforge.enseeiht.fr},
  version 2.4.9} 
code generator.
Our test case chain generator
\textsc{ChainCover}\footnote{\url{http://www.cprover.org/chaincover/},
  version 0.1} itself is built upon the infrastructure provided by 
\textsc{Cbmc}\footnote{\url{http://www.cprover.org/cbmc/}, version 4.4} \cite{CKL04}
with \textsc{MiniSat}\footnote{\url{http://minisat.se}, version 2.2.0} as a SAT backend,
the \textsc{Lkh} TSP
solver\footnote{\url{http://www.akira.ruc.dk/~keld/research/LKH/},
  version 2.0.2}
\cite{Hel00}, and the 
\textsc{Clingo} ASP
solver\footnote{\url{http://potassco.sourceforge.net/},
  version 3.0.5}
\cite{GKK+11}.

%%%%%%%%%%%%%%%%%%%%%%%%%%%%%%%%%%%%%%%%%%%%%%%%%%%%%%%%%%%%%%%%%%%%%%%%%%%%%%%%
\begin{figure}[t]
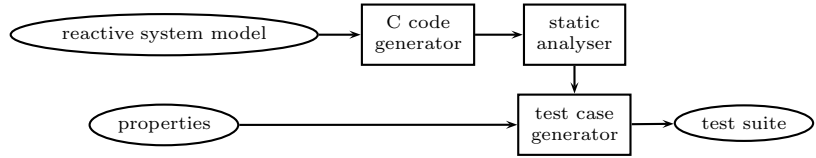

\centering
\vspace*{-2ex}
\scriptsize
\psset{arrows=->}
%\vspace*{1ex}
\begin{tabular}{c@{\hspace{2em}}c@{\hspace{2em}}c@{\hspace{2em}}c}
\ovalnode{sm}{reactive system model} &
\rnode{cg}{\psframebox{$\begin{array}{c}\text{C code}\\ \text{generator}\end{array}$}} &
\rnode{sa}{\psframebox{$\begin{array}{c}\text{static}\\ \text{analyser}\end{array}$}} \\[6ex]
\ovalnode{pr}{properties} && \rnode{tg}{\psframebox{$\begin{array}{c}\text{test case}\\ \text{generator}\end{array}$}} &
\ovalnode{tc}{test suite}
\end{tabular}
\ncline{sm}{cg}
\ncline{cg}{sa}
\ncline{sa}{tg}
\ncline{pr}{tg}
\ncline{tg}{tc}
\vspace*{1ex}
\caption{\label{fig:tool1} Tool chain
}
\end{figure}
%%%%%%%%%%%%%%%%%%%%%%%%%%%%%%%%%%%%%%%%%%%%%%%%%%%%%%%%%%%%%%%%%%%%%%%%%%%%%%%%

The properties are written in C\rronly{ using the \texttt{assert} and
\texttt{\_\_CPROVER\_assume} macros}. 
For instance, property $p_1$ in our example is stated as follows:
{\oldfootnotesize
\begin{lstlisting}
void p_1(t_input* i, t_state* s) {
  __CPROVER_assume(s->mode==ON && s->speed==1 && i->dec);
  compute(i,s);
  assert(s->speed==1);
}
\end{lstlisting}}
Assumptions on the inputs and the state invariant obtained from
the static analysis are written as C code in a similar way.

%-------------------------------------------------------------------------------
\paragraph{Benchmarks}
Our experiments are based on \textsc{Simulink} models, mainly from automotive industry.
\rronly{For some benchmarks, we had the \textsc{Simulink} models or at least
the generated C code available;
for others we only had screenshots from the
\textsc{Simulink} models, which we had to re-engineer ourselves. } %
Our benchmarks are
a simple \emph{cruise} control model \cite{Bos07},
a \emph{window} controller\footnote{\url{http://www.mathworks.co.uk/products/simulink/examples.html}},
a car \emph{alarm} system\footnote{\url{http://www.mogentes.eu/public/deliverables/}\\\url{MOGENTES_3-15_1.0r_D3.4b_TestTheories-final_main.pdf}},
an elevator model \cite{MS11}, and
a model of a \emph{robot arm} that can be controlled with a joystick.
We generated test case chains for these examples for specifications of
different size and granularity.
The benchmark characteristics are listed in Table~\ref{tab:results}.
Apart from \emph{Cruise 1} all specifications have properties with 
multi-state assumptions, thus, the obtained test case chains are not
minimal in general. All our benchmarks are (almost) strongly connected (some
have an initial transition after which the system is strongly
connected), hence, they did not require abstraction refinement.

%%%%%%%%%%%%%%%%%%%%%%%%%%%%%%%%%%%%%%%%%%%%%%%%%%%%%%%%%%%%%%%%%%%%%%%%%%%%%%%%
\begin{table}[t]
\centering
\begin{tabular}{|l|c|c|c|c|c|r@{}l|c|c|r@{}l|c|c|r@{}l|c|c|r@{}l|c|c|}
\hline
 & \multicolumn{3}{|c|}{size}
& \multicolumn{4}{|c|}{\textsc{ChainCover}} 
& \multicolumn{4}{|c|}{\textsc{FShell}} 
& \multicolumn{4}{|c|}{random} 
& \multicolumn{4}{|c|}{\textsc{KLEE}} 
\\
\hline
benchmark  & s     & i  & P  & tcs & len & \multicolumn{2}{|c|}{time} &  tcs & len & \multicolumn{2}{|c|}{time} &  tcs & len & \multicolumn{2}{|c|}{time}  & tcs & len & \multicolumn{2}{|c|}{time} \\\hline
Cruise 1   & 3b    & 3b &  4 & 1 &  9 & 0&.77 & 3 & 18 & 3&.67 & 2.8 &24.6 & 0&.54 & 3 & 27 & 46&.5\\
Cruise 2   & 3b    & 3b &  9 & 1 & 10 & 0&.71 & 4 & 20 & 3&.56 & 2.4&21.2 & 0&.07 & 3 & 30 & 17&.7 \\
\hline
Window 1   & 3b+1i & 5b &  8 & 1 & 24 & 14&.1 & 4 & 32 & 19&.0 & 1.8 &40.4 & 58&.9& 3 & 72 & 155&\\
Window 2   & 3b+1i & 5b & 16 & 1 & 45 & 24&.9  & 7 & 56 & 28&.3 & 2.0 &86.8 & 18&.7& 5 & 225 & 242&\\
\hline
Alarm 1    & 4b+1i & 2b &  5 & 1 & 26 & 7&.51 & 1 & 27 & 509&& \multicolumn{2}{|c|}{80\% cov.} & \multicolumn{2}{|c|}{t/o} & \multicolumn{2}{|c|}{60\% cov.} & \multicolumn{2}{|c|}{t/o} \\
Alarm 2    & 4b+1i & 2b & 16 & 1 & 71 & 33&.5& 3 & 81 & 690&&  \multicolumn{2}{|c|}{94\% cov.} & \multicolumn{2}{|c|}{t/o}& \multicolumn{2}{|c|}{63\% cov.} & \multicolumn{2}{|c|}{t/o} \\
\hline
Elevator 1 & 6b    & 3b &  4 & 1 &  8 & 22&.9 & 2 & 15 & 115&& 2.2 &10.4 & 0&.85 & 2 & 16 & 24&.4 \\
Elevator 2 & 6b    & 3b & 10 & 1 & 32 & 97&.3& 5 & 54 & 789& & 2.6 & 49.0 & 65&.8& \multicolumn{2}{|c|}{70\% cov.} & \multicolumn{2}{|c|}{t/o} \\
Elevator 3 & 6b    & 3b & 19 & 1 & 48 & 458&& 6 & 54 & 838&& 4.0& 149 & 18&.0& \multicolumn{2}{|c|}{53\% cov.} & \multicolumn{2}{|c|}{t/o} \\
\hline
Robotarm 1 & 4b+2f & 3b &  4 & 1 & 25  & 185&& 2 & 22 & 362&& 2.4 &49.0 & 0&.07 & 2 & 40 & 10&.9 \\
Robotarm 2 & 4b+2f & 3b & 10 & 1 & 47   & 113&& 2 & 33& 532&& 3.8 & 72.2 & 0&.21& \multicolumn{2}{|c|}{80\% cov.} & \multicolumn{2}{|c|}{t/o} \\
Robotarm 3 & 4b+2f & 3b & 18 & 1 &  84  & 427&& 5 & 55 & 731&& 3.2 & 160 & 0&.62& \multicolumn{2}{|c|}{67\% cov.} & \multicolumn{2}{|c|}{t/o} \\
\hline
\end{tabular}~\\[2ex]
\caption{\label{tab:results}
{\normalfont
Experimental results: The table lists the number of test cases/chains
(tcs), the accumulated length of the test case chains (len), and the
time (in seconds) taken for test case generation. Size indicates the
size of the program in the number of (minimally encoded) Boolean (b), integer
(i) and floating point (f) variables and
(minimally encoded) Boolean (b)
inputs. ``P'' is the number of properties
in the specification. If the tool timed out (``t/o'') after
1 hour the achieved coverage (``cov'') is given.
}}
\end{table}
%%%%%%%%%%%%%%%%%%%%%%%%%%%%%%%%%%%%%%%%%%%%%%%%%%%%%%%%%%%%%%%%%%%%%%%%%%%%%%%%

%-------------------------------------------------------------------------------
\paragraph{Comparison}
We have compared our tool \textsc{ChainCover} (using \textsc{Lkh}) with
\begin{compactitem}
\item \textsc{FShell}\footnote{\url{http://forsyte.at/software/fshell/}} 
\pponly{\cite{HSTV08}}\rronly{\cite{HSTV08,HSTV09}}, an efficient test generator
with test suite minimisation, 
\item an in-house, simple \emph{random} case generator with test suite
  minimisation, and
\item \textsc{Klee}\footnote{\url{http://klee.llvm.org/}}
  \cite{CDE08}, a test case generator based on symbolic execution.
\end{compactitem}
In order to make results comparable, we have chosen $\sfinal$ to be
equivalent to $\sinit$ (or the state after the initial transition).
Hence, test cases generated by \textsc{FShell}, random, and
\textsc{Klee} can be concatenated (disregarding the initial
transition) to get a single test case chain.

Like our tool, \textsc{FShell} is based on bounded model
checking.  \textsc{FShell} takes a coverage specification in form of a
query as input.  It computes test cases
that start in $\sinit$, cover one or more properties
$p_1,\ldots,p_n$ and terminate in $\sfinal$ when given the query:
\texttt{cover (@CALL(p\_1) | ... | @CALL(p\_n)) -> @CALL(final)}.  In the
best case, \textsc{FShell} returns a single test case, \ie~a test
chain.  We have run \textsc{FShell} with increasing unwinding bounds
$K$ until all properties were covered.

%%%%%%%%%%%%%%%%%%%%%%%%%%%%%%%%%%%%%%%%%%%%%%%%%%%%%%%%%%%%%%%%%%%%%%%%%%%%%%%%
\begin{figure}[t]
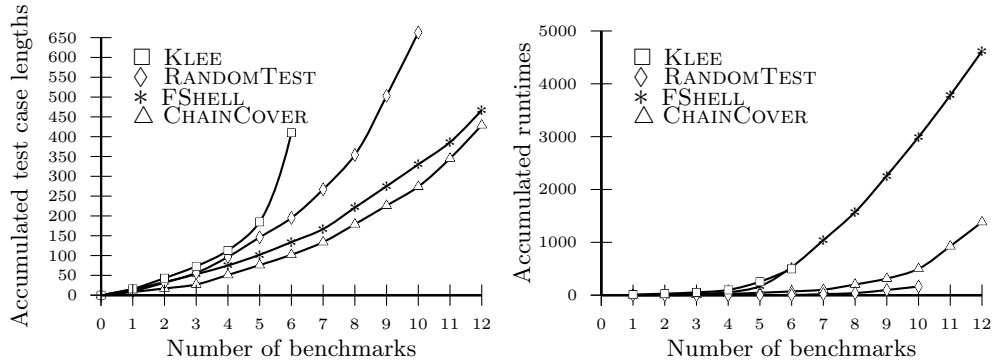

\vspace*{10ex}
\hspace{2em}
\begin{minipage}{0.4\textwidth}
\vspace*{15ex}
\psset{xunit=12pt,yunit=0.15pt}
\def\psxlabel#1{\oldfootnotesize #1}
\def\psylabel#1{\oldfootnotesize #1}
\savedata{\tchain}[
{{0,0},{1, 8}, {2, 17}, {3, 27}, {4,51}, {5, 76},
{6, 102}, {7, 134}, {8, 179}, {9, 226},{10,274}, {11,345}, {12,429}}]
\savedata{\fshell}[
{{0, 0}, {1,15}, {2, 33}, {3, 53}, {4, 75},{5,102}, 
{6,134}, {7,167}, {8,221},{9,275}, {10,330},{11,386},{12,467}}]
\savedata{\random}[
{{0, 0}, {1,10.4}, {2, 31.6}, {3, 56.2}, {4,96.6},{5,146}, 
{6,195}, {7,267}, {8,354},{9,503},{10,663}}]
\savedata{\klee}[
{{0, 0}, {1,16}, {2, 43}, {3, 73}, {4,113},{5,185}, 
{6,410}}]
\dataplot[plotstyle=curve,showpoints=true,
dotstyle=triangle]{\tchain}
\dataplot[plotstyle=curve,showpoints=true,
dotstyle=asterisk]{\fshell}
\dataplot[plotstyle=curve,showpoints=true,
dotstyle=diamond]{\random}
\dataplot[plotstyle=curve,showpoints=true,
dotstyle=square]{\klee}
\psaxes[Dy=50,dy=50,labelFontSize=\scriptstyle](0,0)(12,650)
\rput(6,-20pt){\oldfootnotesize Number of benchmarks}
\rput{90}(-2.5,325){\oldfootnotesize Accumulated test case lengths}
\rput[l](1,600){\footnotesize $\Box\,$ \textsc{Klee}}
\rput[l](1,550){\footnotesize $\Diamond\;$ \textsc{RandomTest}}
\rput[l](1,500){{\large $\ast$}\footnotesize$\,$ \textsc{FShell}}
\rput[l](1,450){\footnotesize $\bigtriangleup$ \textsc{ChainCover}}
\vspace*{5ex}
\end{minipage}
\hspace{5em}
\begin{minipage}{0.4\textwidth}
\vspace*{15ex}
\psset{xunit=12pt,yunit=0.02pt} 
%\psset{xunit=12pt,yunit=24pt} 
\def\psxlabel#1{\oldfootnotesize #1}
\def\psylabel#1{\oldfootnotesize #1}
\savedata{\tchain}[
{{1,0.77}, {2,1.48}, {3,9}, {4,23.1}, {5,46},
{6,70.9}, {7,104}, {8,202}, {9,315},{10,500}, {11,927}, {12,1385}}]
\savedata{\fshell}[
 {{1,3.56}, {2,7.23}, {3,26.2}, {4,54.5},{5,170}, 
 {6,532}, {7,1041}, {8,1573},{9,2263}, {10,2994},{11,3783},{12,4621}}]
\savedata{\random}[
 {{1,0.07}, {2,0.14}, {3,0.35}, {4,0.89},{5,1.5}, 
 {6,2.36}, {7,20.4}, {8,39.1},{9,98},{10,164}}]
\savedata{\klee}[
{{1,10.9}, {2, 28.6}, {3, 53}, {4,99.5},{5,255}, 
{6,497}}]
%\pstScalePoints(1,1){}{log}
%\psaxes[labelFontSize=\scriptstyle,ylogBase=10](0,0)(12,4) 
\psaxes[labelFontSize=\scriptstyle,Dy=1000](0,0)(12,5000) 
\dataplot[plotstyle=curve,showpoints=true,dotstyle=triangle]{\tchain}
\dataplot[plotstyle=curve,showpoints=true,dotstyle=asterisk]{\fshell}
\dataplot[plotstyle=curve,showpoints=true,dotstyle=diamond]{\random}
\dataplot[plotstyle=curve,showpoints=true,dotstyle=square]{\klee}
%\rput{90}(-2.5,2){\oldfootnotesize Accumulated runtimes}
\rput{90}(-2.5,2500){\oldfootnotesize Accumulated runtimes} 
\psset{xunit=12pt,yunit=0.2pt}
\rput(6,-20pt){\oldfootnotesize Number of benchmarks}
\psset{xunit=12pt,yunit=0.15pt}
\rput[l](1,600){\footnotesize $\Box\,$ \textsc{Klee}}
\rput[l](1,550){\footnotesize $\Diamond\;$ \textsc{RandomTest}}
\rput[l](1,500){{\large $\ast$}\footnotesize$\,$ \textsc{FShell}}
\rput[l](1,450){\footnotesize $\bigtriangleup$ \textsc{ChainCover}}
\vspace*{5ex}
\end{minipage}
\caption{\label{fig:results}
Experimental results: accumulative graph of test case lengths on the
left-hand side, accumulated runtimes on the right-hand side.
}
\end{figure}
%%%%%%%%%%%%%%%%%%%%%%%%%%%%%%%%%%%%%%%%%%%%%%%%%%%%%%%%%%%%%%%%%%%%%%%%%%%%%%%%

For random testing and \textsc{Klee}, we coded the requirement to
finish a test case in $\sfinal$ with the help of flags in the test
harness.
Then we stopped the tools as soon as full coverage was achieved and
selected the test cases achieving full coverage while minimising the
length of the input sequence using an in-house, 
weighted-minimal-cover-based test suite minimiser.
For random testing we averaged the results over five runs.
Unlike \textsc{ChainCover} and \textsc{FShell}, which start test 
chain computation without prior knowledge of how many
steps are needed to produce a test case, we had to provide random testing and
\textsc{Klee}  with this information.
The reason is that the decision when a certain number of steps will not
yield a test case can only be taken after reaching a timeout for random
testing.  Similarly, \textsc{Klee} may take hours to terminate. 
Consequently, the results for random testing and \textsc{Klee} are not fully
comparable to those of the other tools.

%-------------------------------------------------------------------------------
\paragraph{Results}
Experimental results obtained are shown in Table~\ref{tab:results} and
Fig.~\ref{fig:results}. 
\begin{compactitem}
\item Our tool \textsc{ChainCover} usually succeeds
  in finding shorter test case chains than the other tools. It is also in general
  faster. \textsc{ChainCover} spends more than 99\% of its runtime with BMC.
  The time for solving the ATSP problem is neglible for the number of
  properties we have in the specifications.  The runtime ratio for
  generating the property $K$-reachability graph ($\mathcal{O}(Kn^2)$
  BMC queries for $n$ properties) versus finding and repairing a chain
  ($\mathcal{O}(Kn)$ BMC queries) varies between 7:92 and 75:24.
\item \textsc{FShell} comes closest to \textsc{ChainCover} with respect to
  test case chain length, and finds shorter chains on the robot arm example. 
  However, \textsc{FShell} takes much longer: the computational
  cost depends on the number of unwindings and the size of the
  program and less on the number of properties.
\item  Random testing yields very good results on some (small)
  specifications and sometimes even finds chains that are as short as those
  generated by \textsc{ChainCover}.  However, the results vary and heavily
  depend on the program and the specification: in some cases, \eg~\emph{Robotarm},
  full coverage is achieved in fractions of a second; in other cases, full
  coverage could not be obtained before reaching the timeout of one hour and
  generating millions of test cases.
\item \textsc{Klee} found test case chains on a few of the benchmarks in
  very short time, but did not achieve full coverage within an hour on
  half of the benchmarks, which suggests that exhaustive
  exploration is not suitable for our problem.
\end{compactitem}

%===============================================================================
\section{Related Work}\label{sec:rw}
%===============================================================================

Test case generation with model checkers came up in the mid-90s and
has attracted continuous research interest since then, especially due
to the enormous progress in SAT solver performance. 
There is a vast literature on this topic, surveyed in~\cite{FWA09},
for example.
\rronly{
The \textsc{FShell} tool~\cite{HSTV09,HSTV08} we have compared with was
developed with the motivation of enabling the flexible specification of the
desired coverage.  
}

\rronly{
%-------------------------------------------------------------------------------
\paragraph{Reactive system testing}
There are many approaches to reactive system testing:
While random testing \rronly{\cite{DN84}} is still commonly used,
approaches have been developed that combine random testing
with \emph{symbolic and concrete execution} \pponly{\cite{GKS05,SA06,CDE08}}\rronly{
(\textsc{Dart}
\cite{GKS05}, \textsc{Cute} \cite{SA06}, \textsc{Klee} \cite{CDE08})}
to guide exhaustive path enumeration.
\emph{Scenario-based testing} employ test specifications to guide test
case generation towards a particular functionality
\pponly{\cite{BORZ99,JRB06,RRJ08}}\rronly{(\eg, \textsc{Lutess} \cite{BORZ99}, \textsc{Lurette} \cite{JRB06}, \textsc{Lutin}
\cite{RRJ08})}.  These methods restrict the input space using
static analysis and apply (non-uniform) random test case generation.
\emph{Model-based testing} (see \cite{PSM12,LY96a} for surveys on this
topic) considers specification models based on labelled transition
systems. For instance, extended finite state machines (EFSM)
\pponly{\cite{UY91}}\rronly{\cite{LY96b,UY91,PBG04}} 
are commonly used in communication protocol
testing to provide exhaustive test case generation for conformance
testing\pponly{\cite{JJ05,Tre08}}. \rronly{Available tools include, \eg, \textsc{Tgv} \cite{JJ05} and
\textsc{TorX} \cite{Tre08}.}
}

\rronly{
%-------------------------------------------------------------------------------
\paragraph{Minimal checking sequences and test optimisation}
}
In the model-based testing domain, the problem of finding minimal checking
sequences has been studied in \emph{conformance
testing}~\pponly{\cite{NMHN13,PSY12,HU10}}\rronly{\cite{NMHN13,PSY12,HU10,HU06,Hie04}}, which amounts to checking
whether each state and transition in a given EFSM specification is correctly
implemented.
First, a minimal checking path is computed, which might be infeasible due to
the operations on the data variables.  Subsequently, random test case
generation is applied to discover such a path, which might fail again. 
Duale and Uyar~\cite{DU04} propose an algorithm for finding a feasible
transition path, but it requres guards and assignments in the models to be
linear.  Another approach is to use genetic algorithms~\cite{NMHN13,KHS09}
to find a feasible path of minimised length.
\rronly{
Also in our setting, the use of genetic algorithms in order to find minimised
instead of minimal solutions is an option to consider.
}
SAT solvers have also been used to compute (non-minimal) checking sequences
in FSM models \cite{JUYZ09,MOF+03}.
Our method does not impose restrictions on guards
and assignments and implicitly handles low-level issues such as
overflows and the semantics of floating-point arithmetic in finding
feasible test cases.
The fact that minimal paths on the abstraction might not be feasible 
in the concrete program does not arise due to limited reasoning
about data variables, but due to the multi-state nature of the 
properties we are trying to cover. 

Closest to our work is recent work~\cite{PVL11} on generating test chains for EFSM models 
with timers. They use SMT solvers to find a path to the nearest test
goal and symbolic execution to constrain the search space. If no test
goal is reachable they backtrack to continue the search from an
earlier state in the test chain.
Their approach represents a greedy heuristics and thus makes
minimality considerations difficult. Our method can handle timing
information if it is explicitly expressed as counters in the program.

Petrenko et al~\cite{PDRM13} propose a method for test optimisation for EFSM models
with timers. They use an ATSP solver to find an optimal ordering of a given set
of test cases and an SMT solver to determine paths connecting them.
The problem they tackle is easier than ours because they do not
generate test cases, but just try to chain a given set of test cases
in an optimal order. Additionally, they take into account overlappings of test
cases during optimisation.

In contrast to all these works, our approach starts from a partial specification given by 
 a set of properties, usually formalised from high-level requirements. 
The $K$-reachability graph abstraction can be viewed as the generation of a model from a
partial specification and automated annotation of model transitions
with timing information in terms of the minimal number of steps required.

%===============================================================================
\section{Summary and Prospects}\label{sec:concl}
%===============================================================================

We have presented a novel approach to discovering a minimal
test case chain, \ie, a single test case that covers a given set of test goals in a
minimal number of execution steps.
Our approach combines reachability analysis to build an abstraction,
TSP-based optimisation and heuristics to find a concrete solution
in case we cannot guarantee minimality.
The test goals might also be generated from an EFSM specification or from code
coverage criteria like MC/DC.  This flexibility is a distinguishing
feature of our approach that makes it equally applicable to
model-based and structural coverage-based testing.
In our experimental evaluation, we have shown that
our tool \textsc{ChainCover} outperforms
state-of-the-art test suite generators.
\rronly{Moreover, our approach is not restricted to C code generated from
\textsc{Simulink}---it can be applied to any reactive system
language.  For instance, we could also consider consider Verilog, or
the application to HW/SW-co-verification combing Verilog and C code.}

%-------------------------------------------------------------------------------
\paragraph{Prospects}
\pponly{ 
In \S\ref{sec:compl} we have proposed an abstraction
refinement method in the case of multi-state property triggers.  The
fundamental problem is that a failed path represents information
about at least two edges that we cannot encode as an equivalent TSP
because it requires side conditions such as the solution not
containing a set of subpaths.
Since our experimental results suggest that the bottleneck of the
approach lies rather in solving reachability queries than TSPs, we
could also opt for using answer set programming (ASP) solvers (\eg~\cite{GKK+11}), which
are less efficient in solving TSPs, but they allow us to specify
arbitrary side conditions.
}%
\rronly{Deep loops pose a problem for BMC-based methods.  For instance,
we had to reduce size of loop bound constants in the car \emph{alarm}
system benchmark to make it tractable for comparison.
Acceleration methods, \eg~\cite{KLW13}, are expected to remedy many such
situations, especially those involving counters.}

\rronly{Moreover, the property $K$-reachability graph generation lends
itself to parallellisation.  This is expected to give a further boost to the
capacity of our tool.}

Test case chains are intended to demonstrate conformance in late
stages of the development cycle, especially in acceptance tests when
the system can be assumed stable.
It is an interesting question in how far they can be used in earlier
phases: The test case chains computed by our method are able to
continue to the subsequent test goals even if a test fails, as long as
the implementation has not changed too much; otherwise the test chain
has to be recomputed.  In this case, it would be desirable to
incrementally adapt the test case chain after bug fixes and code
changes.

%-------------------------------------------------------------------------------
\bigskip\paragraph{\textsc{Acknowledgements}}
We thank Cristian Cadar for his advice regarding the
comparison with \textsc{Klee}, and the anonymous reviewers for their
invaluable comments.

\bibliographystyle{splncs}
\bibliography{biblio}

\end{document}